\documentclass{article}
\usepackage{fullpage,nicefrac,comment}
\usepackage{amssymb,amsmath,amsthm,amsfonts,enumitem, overpic}
\usepackage[margin=1.125in]{geometry}
\usepackage{xcolor}
\definecolor{DarkGreen}{rgb}{0.1,0.5,0.1}
\definecolor{DarkRed}{rgb}{0.5,0.1,0.1}
\usepackage[T1]{fontenc}
\usepackage{amsbsy}
\definecolor{DarkBlue}{rgb}{0.1,0.1,0.5}

\usepackage[small]{caption}
\usepackage{tikz}
\usetikzlibrary{arrows,decorations.pathmorphing,decorations.shapes,snakes,patterns,positioning}
\usepackage[ruled,linesnumbered]{algorithm2e}

\newcommand{\cA}{\ensuremath{\mathcal{A}}}
\newcommand{\cX}{\ensuremath{\mathcal{X}}}

\newcommand{\R}{{\mathbb R}}
\newcommand{\RR}{{\mathbb R}}

\newcommand{\gN}{\mathcal{N}} 

\newcommand{\ones}{\mathbf{1}}
\newcommand{\bX}{\mathbf{X}}
\newcommand{\bB}{\mathbf{B}}
\newcommand{\bA}{\mathbf{A}}
\newcommand{\bH}{\mathbf{H}}
\newcommand{\bY}{\mathbf{Y}}
\newcommand{\bZ}{\mathbf{Z}}
\newcommand{\bS}{\mathbf{S}}
\newcommand{\bU}{\mathbf{U}}
\newcommand{\bV}{\mathbf{V}}
\newcommand{\bF}{\mathbf{F}}
\newcommand{\bM}{\mathbf{M}}
\newcommand{\bQ}{\mathbf{Q}}
\newcommand{\bW}{\mathbf{W}}
\newcommand{\bI}{\mathbf{I}}
\newcommand{\bL}{\mathbf{L}}
\newcommand{\bR}{\mathbf{R}}
\newcommand{\bD}{\mathbf{D}}
\newcommand{\bv}{\mathbf{v}}
\newcommand{\bu}{\mathbf{u}}
\newcommand{\bh}{\mathbf{h}}
\newcommand{\bg}{\mathbf{g}}
\newcommand{\ba}{\mathbf{a}}
\newcommand{\bb}{\mathbf{b}}
\newcommand{\bw}{\mathbf{w}}
\newcommand{\bx}{\mathbf{x}}

\newcommand{\be}{\mathbf{e}}
\newcommand{\bl}{\ensuremath{\boldsymbol\ell}}
\newcommand{\br}{\mathbf{r}}
\newcommand{\bz}{\mathbf{z}}
\newcommand{\bxi}{\boldsymbol{\xi}}

\newcommand{\PR}[1]{{\mathbb{P}}\left\{ #1\right\}}
\newcommand{\E}[1]{\mathbb{E}\left[ #1\right]}
\newcommand{\EE}{\mathbb{E}}

\newcommand{\norm}[1]{\left\|#1\right\|}
\newcommand{\twonorm}[1]{\left\|#1\right\|_2}
\newcommand{\onenorm}[1]{\left\|#1\right\|_1}
\newcommand{\infnorm}[1]{\left\|#1\right\|_\infty}
\newcommand{\opnorm}[1]{\left\|#1\right\|}
\newcommand{\nucnorm}[1]{\left\|#1\right\|_*}
\newcommand{\decnorm}[2]{\left\|#1\right\|_{#2}}
\newcommand{\fronorm}[1]{\decnorm{#1}{F}}
\newcommand{\maxnorm}[1]{\decnorm{#1}{\max}}

\newcommand{\inabs}[1]{\left|#1\right|}
\newcommand{\KL}[2]{D_{KL}\left(#1\, \|\, #2 \right)}
\newcommand{\ip}[2]{\ensuremath{\left\langle #1,#2\right\rangle}}

\newcommand{\inset}[1]{\left\{#1\right\}}

\newcommand{\inparen}[1]{\left(#1\right)}

\newcommand{\suchthat}{\,:\,}

\newcommand{\had}{\circ}

\newcommand{\argmin}{\mathrm{argmin}}

\newcommand{\rank}{\ensuremath{\operatorname{rank}}}

\newcommand{\tr}{\mathrm{tr}}

\newcommand{\ind}[1]{\ensuremath{\mathbf{1}_{#1}}}

\newcommand{\eps}{\varepsilon}
\renewcommand{\epsilon}{\varepsilon}

\newtheorem{theorem}{Theorem} 
\newtheorem{lemma}[theorem]{Lemma}

\newtheorem{corollary}[theorem]{Corollary} 

\newtheorem{remark}[theorem]{Remark}
\newtheorem{claim}[theorem]{Claim}

\newtheorem{question}[theorem]{Question}
\newtheorem{example}{Example}

\pdfoutput=1
\graphicspath{{figs/}}
\title{
Weighted matrix completion from non-random, non-uniform sampling patterns}

\author{Simon Foucart, Deanna Needell, Reese Pathak, Yaniv Plan, Mary Wootters}
\begin{document}
\maketitle
\begin{abstract}
We study the matrix completion problem when the observation pattern is deterministic and possibly non-uniform. We propose a simple and efficient debiased projection scheme for recovery from noisy observations and analyze the error under a suitable weighted metric. We introduce a simple function of the weight matrix and the sampling pattern that governs the accuracy of the recovered matrix. We derive theoretical guarantees that upper bound the recovery error and nearly matching lower bounds that showcase optimality in several regimes. Our numerical experiments demonstrate the computational efficiency and accuracy of our approach, and show that debiasing is essential when using non-uniform sampling patterns.
\end{abstract}

\section{Introduction}\label{sec:intro}
The \textit{matrix completion problem} is to determine a complete $d_1 \times d_2$ matrix from a subset $\Omega \subset [d_1] \times [d_2]$ of its entries.  A typical assumption that makes such a problem well-posed is that the underlying matrix from which the entries are observed is low-rank (or approximately low-rank).  Matrix completion has many applications, including collaborative filtering \cite{goldberg1992using}, system identification \cite{liu2009interior}, sensor localization \cite{biswas2006semidefinite,singer2009remark,singer2010uniqueness}, rank aggregation \cite{gleich2011rank}, scene recovery in imaging \cite{chen2004recovering,tomasi1992shape}, multi-class learning \cite{abernethy2006low,amit2007uncovering,argywho}, and more.
  This is now a well-studied problem, and there are several main approaches to its solution, such as low-rank projection \cite{keshavan2010matrix, keshavan2010noisy} and convex optimization \cite{srebro2004maximum, candes2010matrix}, which have rigorous provable recovery guarantees (see e.g. \cite{gross2011recovering, candes2009exact, candes2010power, keshavan2010matrix, keshavan2010noisy, negahban2010restricted, koltchinskii2011nuclear, candes2010matrix, recht2011simpler, rohde2011estimation, klopp2011rank, gaiffas2010sharp, klopp2011high, koltchinskii2012neumann}).  
  
Besides a low-rank assumption on the underlying matrix, one also clearly needs to assume something on the sampling pattern $\Omega$.  Theoretical guarantees for matrix completion typically enforce that the sampling pattern is obtained from  (most often uniform) random sampling~\cite{meka2009matrix, srebro2010collaborative, klopp2014noisy,bhojanapalli2014tighter,eftekhari2016mc,eftekhari2016weighted}.  
However, for many applications, the sampling pattern may not be uniformly random, and indeed may not be reasonably modeled as random at all.

In this paper, we study the 
problem of matrix completion with \em deterministic \em sampling.  This version of the problem has been studied before~\cite{lee2013matrix,heiman2014deterministic,Bhojanapalli2014, li2016recovery,debiasFNPW17}, although much less extensively than the case with random sampling. 
Our goal in this paper is to address the following question:
\begin{question}\label{q:intro} Given a deterministic sampling pattern $\Omega$ and corresponding (possibly noisy) observations from the matrix, what type of recovery error can we expect, in what metric, and how may we efficiently implement this recovery? \end{question} 

We formalize the question above and propose an appropriate \textit{weighted} error metric for matrix recovery that will take the form $\|\bH\had (\hat{\bM} - \bM)\|_F$, where $\bM$ is the true underlying low-rank matrix, $\hat{\bM}$ the recovered matrix, and $\bH$ an appropriate weighting matrix.   Most prior work that assumes uniform sampling of matrix entries considers an unweighted error metric.  This corresponds with taking $H_{i,j} = 1$, thus $H$ is a rank-1 matrix.  We show that matrix completion theory can be extended to many non-uniform, non-random cases by allowing $H$ to be an (almost) arbitrary rank-1 matrix appropriately approximating the sampling pattern.  When $\bH$ satisfies certain conditions with respect to the sampling pattern $\Omega$, which can be easily verified, we show that a simple ``debiased projection'' algorithm performs well.  Moreover, this algorithm is extremely efficient.  
We also establish lower bounds that show that our debiased projection algorithm is nearly optimal in several situations.  Finally, we include numerical results that demonstrate the efficacy and advantages of our approach. 

\subsection{Background and motivation}
Given a sampling pattern $\Omega$, we write $\ind{\Omega}$ to denote the matrix whose entries are $1$ on $\Omega$ and zero elsewhere, so that the entries of $\bM_{\Omega} = \ind{\Omega}\had\bM$ are equal to those of $\bM$ on $\Omega$, and are equal to $0$ on $\Omega^c$.  Above, $\had$ denotes the (entrywise) Hadamard product. 

Existing work \cite{{heiman2014deterministic,Bhojanapalli2014,li2016recovery}} shows that when $\ind{\Omega} \in \{0,1\}^{d_1 \times d_2}$ is ``close'' to an appropriately scaled version of the all-ones matrix (more precisely, when $\ind{\Omega}$ is the adjacency matrix of an expander graph) and the matrix $\bM$ is sufficiently incoherent, 
it is possible to efficiently recover an estimate $\hat{\bM}$ so that $\|\hat{\bM} - \bM\|_F$ is small.

Of course, there are some sampling patterns so that we cannot hope to recover $\hat{\bM}$ with small error $\|\hat{\bM} - \bM\|_F$: as a simple (and extreme) example consider sampling the entire \em left \em half of a matrix.  We can exactly recover the left half of $\bM$, but we learn nothing about the right half of $\bM$. 

Thus, our goal will be to find a weight matrix $\bH$ so that we can recover $\hat{\bM}$ so that 
$\fronorm{ \bH \had ( \hat{ \bM } - \bM ) }$ is small.  In the first example where $\ind{\Omega}$ is close to the all-ones matrix, we would choose $\bH$ to be the all-ones matrix; in the second example, we would choose $\bH$ to be the matrix with ones on the left half and zeros on the right half.
More precisely, we begin with the following question, which refines Question~\ref{q:intro}.
\begin{question}\label{q:almost}
Given a sampling pattern $\Omega$, and noisy observations $\bM_\Omega + \bZ_\Omega$, for what weight matrices $\bH$ 
can we efficiently find a matrix $\hat{\bM}$ so that 
$\fronorm{ \bH \had ( \hat{ \bM } - \bM ) }$ is small compared to $\fronorm{\bH}$?  
\end{question}

This is a weighted version of mean squared error.  Indeed, if $H_{i,j} = 1$ for all $i,j$, then the ratio of $\fronorm{ \bH \had ( \hat{ \bM } - \bM ) }^2$ to $\fronorm{\bH}^2$ is precisely the mean squared error per entry.  The notion of using weights in measuring error in this context is not new. For example, in \cite{heiman2014deterministic} the authors use graph sparsifiers to show that given some weight matrix $\bH$, a sampling pattern $\Omega$ can be efficiently found such that  $\fronorm{ \bH \had ( \hat{\bM} - \bM ) }$ is small relative to $\|\bH\|_F$. 
In \cite{lee2013matrix}, the authors give a method for identifying the best weight matrix for this guarantee by observing a real matrix that is close to a given probability distribution on the observed entries. 
However, while these results are extremely general, they have a few drawbacks.  First, these approaches do not tolerate noise in the observed entries.  Second, the algorithms presented in these works may not be efficient: the guarantees are given for the solutions to optimization problems (of the form ``minimize some matrix quantity subject to agreeing with the observed entries'') which may not be easy to solve.  Finally, these works do not prove any lower bounds to show that their guarantees are optimal.  (We discuss these and other related works in more detail in Section~\ref{sec:related}).

In this work, we sacrifice some of the generality of the results discussed above, but we are able to overcome these drawbacks.  More precisely, we will see that if we add the additional restriction that the weight matrix 
 $\bH$ is rank-1, we can answer Question~\ref{q:almost} (in the presence of noise, with efficient algorithms, and with lower bounds that are nearly tight in some cases).  Thus, the question that we address in this work is the following: 

\begin{question}\label{q:main}
Given a sampling pattern $\Omega$, and noisy observations $\bM_\Omega + \bZ_\Omega$, for what \em rank-one \em weight matrices $\bH$ 
can we efficiently find a matrix $\hat{\bM}$ so that 
$\fronorm{ \bH \had ( \hat{ \bM } - \bM ) }$ is small compared to $\fronorm{\bH}$?  And how can we efficiently find such weight matrices $\bH$, or certify that a fixed $\bH$ has this property?
\end{question}

\subsection{Our results}
The goal of this paper is to answer Question~\ref{q:main}.
Our method is a simple weighted (which we refer to as ``debiased") projection algorithm that performs well precisely when the quantity $\lambda = \opnorm{ \bH - \bH^{(-1)} \had \ind{\Omega} }$ is small; here $\bH^{(-1)}$ denotes the Hadamard (entry-wise) inverse.  We note that this parameter $\lambda$ is efficient to compute.
  In addition, we derive lower bounds that show that our algorithm is nearly optimal in several situations.  

To illustrate our upper and lower bounds, we consider two extreme cases, depending on the magnitude of $\lambda$:
\begin{itemize}
\item
First, we consider settings where $\lambda$ is small; as an example, let $(i,j) \in \Omega$ with probability\footnote{Even though for this case study we imagine drawing $\Omega$ at random, our results are still uniform: that is, we draw $\Omega$ once and fix it, and then prove that our algorithm works deterministically for all $\bM$.} $H_{ij}^2$.
In this case, we show that when the error matrix $\bZ$ has standard deviation $\sigma$ of roughly the same order of magnitude as the entries of $\bM$, then the guarantees of our algorithm are nearly tight. 

As an application of this setting, we show how our results can also be used to recover results that are qualitatively similar to the results of~\cite{CBSW15}, 
which show how to do matrix completion for incoherent matrices by sampling proportional to leverage scores.  Unlike that work, our results hold more generally for approximately-low-rank matrices, although as discussed more below the two works are incomparable.
\item
Second, we consider settings where $\lambda$ is large; we focus on symmetric sampling patterns $\Omega$ for which $\ind{\Omega}$ has a large gap between the largest and second-largest eigenvalues.  While we are not able to get the ``correct'' dependence on our parameter $\lambda$ in general, we are able to show that the error must necessarily increase as $\lambda$ increases, which suggest that our approach is qualitatively correct.
\end{itemize}

Finally, we present empirical results using both real and synthetic sampling patterns that show that debiasing is essential when the sampling pattern is non-uniform. We showcase reduction in recovery errors compared with standard (non debiased) approaches as we vary the sampling pattern constructed via regular graphs.

\subsection{Organization}
The rest of the paper is organized as follows. 
In Section \ref{sec:prelims} we instantiate notation, formalize the problem setup, and record several useful results that our theory will utilize. We establish our main results in Section \ref{sec:results} and relate our work to existing work in Section \ref{sec:related}. Sections \ref{sec:upper} and \ref{sec:lower} give generalized upper and lower bounds, respectively, which are specialized to important settings in Sections \ref{sec:nogap} and \ref{sec:gap}. We display numerical results in Section \ref{sec:experiments}.

\section{Set-up and Preliminaries}\label{sec:prelims}
In this section we set notation and our formal problem statement. 
\subsection{Notation}
We begin by setting notation.
For an integer $d$, we use $[d]$ to mean the set $\{1, \ldots, d\}$.  
Throughout the paper, bold capital letters ($\bX$) represent matrices, and bold lowercase letters ($\bx$) represent vectors.  Entries of a matrix $\bX$ or a vector $\bx$ are denoted by $X_{i,j}$ or $x_i$ respectively.  
For a set $\Omega \subseteq [d_1] \times [d_2]$ and a matrix $\bX \in \RR^{d_1 \times d_2}$, we use $\bX_\Omega$ to denote the matrix which is equal to $\bX$ on entries in $\Omega$ and $0$ otherwise.

We use $\infnorm{ \bX } = \max_{i,j} |X_{i,j}|$ to denote the entry-wise $\ell_\infty$ norm for matrices, and $\fronorm{\bX} = \sqrt{ \sum_{i,j} X_{i,j}^2 }$ to denote the Frobenius norm.   We define the \textit{max-norm} by
\[\maxnorm{\bX} := \min_{\bX=\bU\bV^\top}\|\bU\|_{2,\infty}\|\bV\|_{2,\infty}.\]
Letting $B_{\max}$ be the max-norm unit ball, 
$B_{\max} := \{\bX \in \R^{d_1 \times d_2} : \maxnorm{\bX} \leq 1\}$, 
 Grothendieck's inequality \cite[Chapter 10]{jameson1987summing} shows that $B_{\max}$ is close to a polytope with vertices of flat $\rank$-1 matrices.  Concretely, letting $\mathcal{F}$ be the set of flat, $\rank$-1, matrices:
$\mathcal{F} := \{u v^T : u \in \{+1, -1\}^{d_1}, \, v \in \{+1, -1\}^{d_2} \},$
 Grothendieck's inequality states that $
{\rm conv}(\mathcal{F}) \subset B_{\max} \subset K_G \cdot {\rm conv}(\mathcal{F}),
$
where $K_G \leq 1.783$ is Grothendieck's constant.
Given a $\rank$-$r$ matrix $\bM$ with $\infnorm{\bM} \leq \gamma$, we have that 
$\maxnorm{\bM} \leq \sqrt{r} \gamma$ (see \cite[Corollary 2.2]{rashtchian2016bounded}).
In this sense, the max norm serves as a proxy for the rank of a flat matrix that is robust to small perturbations. 

We use $K_r$ to denote the cone of rank-$r$ matrices.
We use $B_\infty$, $B_{\max}$, $B_F$ respectively to denote the unit balls for the corresponding norm.

\subsection{Formal set-up}
We now describe our formal set-up.  
Suppose that $\bM \in \RR^{d_1 \times d_2}$ is a unknown matrix.
In this paper, we will assume either that $\bM$ is in $K_r \cap \beta B_\infty$---that is, $\bM$ is flat and has low rank---or that $\bM$ is in $\beta \sqrt{r} B_{\max}$---that is, $\bM$ is ``approximately'' low-rank.

Fix a sampling pattern $\Omega \subseteq [d_1] \times [d_2]$ and a rank-$1$ matrix $\bW$.\footnote{In the introduction we discussed a rank-$1$ weight matrix $\bH$ which plays the role of $\bW^{(1/2)}$.  Stating the results that way is easier to parse in an introduction, but it will be more convenient for the proofs to state the formal problem in terms of $\bW$.}  Our goal will be to design an algorithm that gives provable guarantees for a worst-case $\bM$, even if it is adapted to $\Omega$.
Our algorithm will observe $\bM_\Omega + \bZ_\Omega$, where $Z_{i,j} \sim \gN(0,\sigma^2)$ are i.i.d. Gaussian random variables.  From these observations, the goal is to learn something about $\bM$.  
Notice that, depending on $\Omega$, it might not be possible to estimate $\bM$ well in a standard metric (like the Frobenius norm). 
Instead, in this paper, we are interested in learning $\bM$ with small error in a \em weighted \em Frobenius norm; that is, we'd like to develop efficient algorithms to find a matrix $\hat{\bM}$ so that
\[ \inparen{ \sum_{i,j} W_{ij} (M_{ij} - \hat{M}_{ij} )^2 }^{1/2} = \fronorm{ \bW^{(1/2)} \had (\bM - \hat{ \bM} ) }\] 
is small for some matrix $\bW$ of interest.  On the other hand,
we will also prove lower bounds that demonstrate for which $(\Omega, \bW)$ combinations certain error bounds Re not possible.

When measuring weighted error, it is important to normalize appropriately in order to understand what the bounds mean.  In our setting, we will always report error normalized by $\fronorm{ \bW^{(1/2)} }$: that is the goal is that 
\[ \frac{ \fronorm{ \bW^{(1/2)} \had ( \bM - \hat{ \bM } )  }}{ \fronorm{ \bW^{(1/2)} } } = \inparen{ \sum_{i,j} \frac{ W_{ij} }{ \sum_{i', j'} W_{i'j'} } (M_{ij} - \hat{M}_{ij})^2 }^{1/2}\]
is small.  Written out this way, it is clear that this gives a weighted average of the per-entry squared error.
In light of the discussion above, we formally define our problem below.
\begin{center}
\fbox{
\parbox{\textwidth}{
\textbf{Problem:} \textsc{Weighted Universal Matrix Completion}

\vspace{.4cm}
\textbf{Parameters:}
\begin{itemize}
\item Dimensions $d_1, d_2$
\item A sampling pattern $\Omega \subset [d_1] \times [d_2]$
\item Parameters $\sigma, \beta, r > 0$
\item A rank-1 weight matrix $\bW \in \RR^{d_1 \times d_2}$ so that $W_{ij} > 0$ for all $i,j$.
\item A set $K$ (which for us will either be $K_r \cap \beta B_\infty$ or $\beta \sqrt{r} B_{\max}$)
\end{itemize}

\textbf{Goal:} Design an efficient algorithm $\cA$ with the following guarantees:
\begin{itemize}
\item $\cA$ takes as input
entries $\bM_\Omega + \bZ_\Omega$ so that $Z_{ij} \sim \gN(0,\sigma^2)$ are i.i.d. 
\item $\cA$ runs in polynomial time
\item With high probability over the choice of $\bZ$, $\cA$ returns an estimate $\hat{\bM}$ of $\bM$ so that
\[ \inparen{  \sum_{i,j} \inparen{ \frac{ W_{ij} }{ \sum_{i', j'} W_{i'j'} } } ( M_{ij} - \hat{M}_{ij} )^2}^{1/2} = \frac{ \fronorm{ \bW^{(1/2)} \had ( \bM - \hat{ \bM } ) }} { \fronorm{ \bW^{(1/2)} } } \leq \delta(d_1, d_2, \Omega, r, \sigma, \beta) \]
for all $\bM \in K$, where $\delta$ is some function of the problem parameters.
\end{itemize}
}
}
\end{center}
Ideally, the error function $\delta$ will tend to zero as $|\Omega|$ grows.

\begin{remark}[Universality]
We emphasize that the requirement in the problem above is a universal one.  That is, for a fixed sampling pattern, the algorithm $\bA$ must work simultaneously for \em all \em relevant matrices~$\bM$.
\end{remark}

\begin{remark}[Strictly positive $\bW$]
Notice that the requirement that $W_{ij}$ be strictly greater than zero (rather than possibly equal to zero) is without loss of generality.  Indeed, if $W_{ij} = 0$ for some $(i,j)$, then either the $i$'th row or the $j$'th column of $\bW$ are zero, and we can reduce the problem to a smaller one by ignoring that row or column.
\end{remark}

\subsection{Useful theorems}
In this section we record a few theorems that we will use in our analysis.
We will use the Matrix-Bernstein Inequality (Theorem 1.4 in \cite{tropp2012}).
\begin{theorem}\label{thm:matrixbernstein}
Let $\bX_i \in \RR^{d \times d}$ for $i=1,\ldots, n$ be independent, random, symmetric matrices, so that
\[ \EE \bX_i = 0 \qquad \text{and} \qquad \opnorm{ \bX_i } \leq R \text{ \ \ \ almost surely. } \]
Then for all $t \geq 0$,
\[ \PR{ \opnorm{ \sum_i \bX_i } \geq t } \leq d \cdot \exp \inparen{ \frac{ - t^2 / 2 }{ \sigma^2 + Rt/3 } } \qquad \text{where} \qquad \sigma^2 = \opnorm{ \sum_i \EE( \bX_i^2 ) }. \]
\end{theorem}

We will also use the following bound about sums of random matrices with Gaussian coefficients (Theorem 1.5 in \cite{tropp2012})
\begin{theorem}\label{thm:matrixgaussian}
Let $\bX_i \in \RR^{d_1 \times d_2}$ for $i = 1, \ldots, n$ be any fixed matrices, and let $g_1, \ldots, g_n$ be independent standard normal random variables.  Define
\[ \sigma^2 = \max \inset{ \opnorm{ \sum_i \bX_i \bX_i^T }, \opnorm{ \sum_i \bX_i^T \bX_i } }. \]
Then for all $t > 0$, 
\[ \PR{ \opnorm{ \sum_i g_i \bX_i } \geq t } \leq (d_1 + d_2) \cdot \exp \inparen{ \frac{ -t^2 }{ 2\sigma^2 } }. \] 
\end{theorem}

We will also need the Hanson-Wright Inequality (see, e.g., \cite{rudelsonhw}).
\begin{theorem}[Hanson-Wright inequality]\label{thm:hansonwright}
There is some constant $c > 0$ so that the following holds.  Let $\bxi \in \{0,\pm 1\}^d$ be a vector with mean-zero, independent entries, and let $\bF$ be any matrix which has zero diagonal.  Then
\[ \PR{ |\bxi^T \bF \bxi| > t } \leq 2 \exp \inparen{ - c \cdot \min \inset{ \frac{ t^2 }{ \fronorm{\bF}^2 }, \frac{ t }{ \opnorm{\bF} } } }. \]
\end{theorem}

\section{Results}\label{sec:results}
In this section, we state informal versions of our results.  We assume that $d_1 = d_2 = d$ to make the results easier to parse, although most of our results extend to rectangular matrices.\footnote{The only exception are Theorems~\ref{thm:gaplower} and \ref{thm:gapupper} which are in terms of the eigenvalues of $\ind{\Omega}$.  For these bounds we assume that $\ind{\Omega}$ is square and symmetric.}  We present more detailed statements of our results later in the paper.

\subsection{General Results}
Our main upper bounds give two algorithms for weighted universal matrix completion.  These are formally stated in Theorems~\ref{lem:upperrankk} and \ref{lem:uppermaxball}, and we give an informal version below.  
Our bounds will depend on two parameters of the sampling pattern $\Omega$ and the weight matrix $\bW$.  For a fixed $\bW$ and $\Omega$, define
\begin{equation}\label{eq:lambda}
 \lambda = \opnorm{ \bW^{(1/2)} - \bW^{(-1/2)} \had \ind{\Omega} } 
\end{equation}
\begin{equation}\label{eq:mu}
 \mu^2 = \max\inset{ \max_i \inparen{ \sum_j \frac{ \ind{ (i,j) \in \Omega } }{ W_{ij} } }, \max_j \inparen{ \sum_i \frac{ \ind{(i,j) \in \Omega} }{ W_{ij} } } }. 
\end{equation}

The parameter $\lambda$ is a measure of how ``close'' $\ind{\Omega}$ is to the matrix $\bW$.  Indeed, if $\ind{\Omega}$ happens to be rank $1$ and $\ind{\Omega} = \bW$, then $\lambda = 0$.

The parameter $\mu$ measures how ``close'' $\ind{\Omega}$ is to $\bW$, as well as capturing how ``lopsided'' they are.  If $\ind{\Omega} = \bW$ then $\mu^2$ is just the max column or row weight of $\Omega$.  However, if $W$ is very different from $\ind{\Omega}$, for example, by putting not very much weight on a row that is heavily sampled by $\Omega$, then $\mu$ will be larger.

We study two algorithms.  The first, which applies when $\bM$ is exactly rank $k$, is a simple debiased projection-based method. 
More precisely, we will estimate a rank-$r$ matrix $\bM$ from the observations $\bY_\Omega = \bM_\Omega + \bZ_\Omega$ by 
\[ \hat{ \bM }_0 = \bW^{(-1/2)} \had \argmin_{ \text{rank}(\bX) = r } \norm{ \bX - \bW^{(-1/2)} \had (\bY_\Omega) }. \]
In Theorem~\ref{lem:upperrankk}, we will show the following.

\begin{theorem}[General upper bound for rank-$k$ matrices, informal]\label{thm:rkkinformal}
Let $\bW \in \RR^{d \times d}$ be a rank-one matrix with strictly positive entries, and fix $\Omega \subseteq [d] \times [d]$.  
Suppose that $\bM \in \RR^{d \times d}$ has rank $r$ and $\infnorm{\bM} \leq \beta$, and let $\bY = \bM + \bZ$ where the entries of $\bZ$ are i.i.d. $\gN(0,\sigma^2)$.
Then with probability at least $1 - 1/d$ over the choice of $\bZ$,
\[ \frac{\fronorm{ \bW^{(1/2)} \had ( \bM - \hat{ \bM }_0 ) }}{ \fronorm{ \bW^{(1/2)} } } \lesssim  \frac{ \beta r \lambda +   \sigma \mu \sqrt{r \log(d)} }{ \fronorm{ \bW^{(1/2)} } }. \]
\end{theorem}

The second algorithm applies when $\bM \in \beta \sqrt{r} B_{\max}$ is approximately low-rank.
Let
\[ \hat{ \bM }_1 = \bW^{(-1/2)} \had \argmin_{ \maxnorm{ \bX } \leq \beta \sqrt{r} } \norm{ \bX - \bW^{(-1/2)} \had (\bM_\Omega + \bZ_\Omega) }. \]

\begin{theorem}[General upper bound for approximately rank-$r$ matrices, informal]\label{thm:appxrkkinformal}
Let $\bW \in \RR^{d \times d}$ be a rank-one matrix with strictly positive entries, and fix $\Omega \subseteq [d] \times [d]$.  Suppose that $\bM \in \RR^{d \times d}$ has $\maxnorm{ \bM } \leq \beta\sqrt{r}$ and let $\bY = \bM + \bZ$ where the entries of $\bZ$ are i.i.d. $\gN(0, \sigma^2)$.
Then with probability at least $1 - 1/d$ over the choice of $\bZ$,
\[ \frac{\fronorm{ \bW^{(1/2)} \had ( \bM - \hat{ \bM }_1 ) }}{\fronorm{ \bW^{(1/2)} } } \lesssim  \sqrt{\beta} \inparen{ \frac{ \beta r \lambda + \sigma \mu \sqrt{ r \log(d) } }{ \fronorm{ \bW^{(1/2)} } } }^{1/2}. \] 
\end{theorem}

Notice that the only difference between the two guarantees is that the average weighted per-entry error bound for approximately low-rank matrices is the square root of the error bound for exactly rank-$r$ matrices.   As we will see below, this translates into the following fact: if we want the average weighted per-entry error to be at most $\eps$, then the number of samples required for the exactly rank-$r$ case will scale like $1/\eps$, while the number of samples required for the approximately rank-$r$ case will scale like $1/\eps^2$.  This quantitative behavior has been observed before (eg, in \cite{1bit-MC}), and we will also show that this dependence on $\eps$ is necessary.

\begin{remark}[Computable Parameters]
We note that both $\lambda$ and $\mu$ are quite easy to compute.  This is valuable because it means that, given a sampling pattern $\Omega$ and a desired weight matrix $\bW^{(1/2)}$, one can quickly compute the guarantees given by Theorems~\ref{thm:rkkinformal} and \ref{thm:appxrkkinformal}.
In contrast, common deterministic conditions which guarantee accurate recovery under random samples (for example the \em restricted eigenvalue condition \em \cite{negahban2010restricted}) are not in general computationally easy to verify.
\end{remark}

The bounds above are a bit difficult to parse: how should we think of $\lambda$ and $\mu$?  As we will see below, it depends on the setting, and in particular on whether or not $\ind{\Omega}$ is ``close'' to $\bW$.  In order to understand the bounds below, we specialize them to two cases.  In the first, we consider $\Omega$ which by construction are ``close'' to $\bW$, so that $\lambda$ is small.  In the second, we consider sampling patterns $\Omega$ that are ``far'' from $\bW$.

\subsection{Case study: When $\lambda$ is small} \label{sec: small lambda}
First, we study the case when  $\lambda$ is small (that is, when $\bW$ is close to $\ind{\Omega}$).
Suppose that $\bW$ has entries in $(0,1]$.
We'd like to consider a ``typical'' $\Omega$ that is close to $\bW$, so we study a random sampling pattern $\Omega$ so that $(i,j) \in \Omega$ with probability $W_{ij}$, independently for each $(i,j)$.  Below, we will use the shorthand ``$\Omega \sim \bW$'' to describe $\Omega$ that is sampled in this way.

We emphasize that even though $\Omega$ is drawn at random in this thought experiment, the goal is to understand our bounds for \em deterministic \em sampling matrices $\Omega$.  That is, the upper bounds are still uniform (they hold simultaneously for all appropriate matrices $\bM$), and this model is just a way to generate matrices $\Omega$ so that $\lambda$ is small, on which to test our uniform bounds.  We will show that for most $\Omega$ that are close to $\bW$ (in the above sense), the upper bound above is nearly tight.  
In this random setting, an essential difference between our results and much of the prior work is computability of parameters.  A key sufficient condition for good matrix completion is the \textit{restricted eigenvalue condition} \cite{negahban2010restricted}, which holds with high probability in the random setting.  However, we believe there are no known polynomial time methods to compute the parameters involved in the restricted eigenvalue condition, and so in general it cannot be verified under model uncertainties.  In contrast, the parameters proposed in this paper are easily computable.

In order to make sure that an $\Omega$ drawn from this ensemble is actually close to $\bW$, we also need to assume that the entries of $\bW$ are not too small; in particular, that they are not smaller than $1/d$; otherwise, it is not hard to see that the parameters $\lambda$ and $\mu$ can become large.  In this model, this means we are assuming that there are at least $\sqrt{d}$ observations in $\Omega$ per row or column.

\begin{remark}[Other places this case is interesting]  Even though the goal of this randomized thought experiment is to analyze deterministic bounds, the upper bounds may be interesting for randomized applications as well.  For example, it is not unreasonable to suppose that in some situations, the entries of $\Omega$ may be drawn independently at random from a distribution that is not uniform.
If this distribution is given by a rank-$1$ matrix, then our results apply. Such may be the case in many applications, for example in survey design where an important batch of questions need to be asked more often than others, or in natural recommender systems where a group of users rate more often than others (or some products are reviewed more than others). 

As an example of how our results for the $\Omega \sim \bW$ case may be interesting in other situations, in Section~\ref{sec:nogapapp} we show how to use our results to recover results similar to those in \cite{CBSW15} about sampling proportional to leverage scores.

\end{remark}

Our results in this setting show that, under these assumptions, our upper bounds above are nearly tight in the setting when $\beta \approx \sigma$ (that is, when the noise is on the same order as the entries of $\bM$).  The formal results are given in Theorems~\ref{thm:upperOmega} and \ref{thm:upperOmegaApprox} (upper bounds for rank $k$ and approximately rank $k$ respectively) and in Theorems~\ref{thm:simOmegaLB} and \ref{lem:flatLBapprox} (lower bounds).  We summarize these informally below.

We begin with our results for exactly rank-$r$ matrices.
\begin{theorem}[Results for rank-$r$ matrices when $\Omega \sim \bW$, Informal]
Let $\bW \in \R^{d \times d}$ be a rank-1 matrix so that for all $i,j$, $1/d \leq W_{ij} \leq 1.$ 
Choose $\Omega \sim \bW$ as described above.

\textbf{Upper bound:}
With probability at least $1 - O(1/d)$ over the choice of $\Omega$, the following holds.
There is an algorithm $\cA$ so that for any rank-r matrix $\bM$ with $\infnorm{\bM} \leq \beta$,
$\mathcal{A}$ returns $\hat{\bM} = \cA( \bM_\Omega + \bZ_\Omega)$
so that with probability at least $1 - 1/d$ over the choice of $\bZ$,
\[
\frac{\fronorm{ \bW^{(1/2)} \had ( \bM - \hat{ \bM } ) }}{ \fronorm{ \bW^{(1/2)} } }  \lesssim  
\sigma  \sqrt{\frac{ rd }{|\Omega|} }\log(d) +
\beta \sqrt{ \frac{ r^2 d } {|\Omega|}} \log(d).
\]

\textbf{Lower bound:}
On the other hand, 
with probability at least $1 - e^{-O( \fronorm{ \bW^{(1/2)} }^2) }$ over the choice of $\Omega$, 
for any algorithm that only sees the values $\bM_\Omega + \bZ_\Omega$ and returns $\hat{\bM}$, there is some rank $r$ matrix $\bM$ with $\infnorm{ \bM } \leq \beta$ so that with probability at least $1/2$ over the choice of $\bZ$, 

\[\frac{ \fronorm{ \bW \had ( \bM - \hat{\bM}  ) }  }{ \fronorm{ \bW^{(1/2)} } }
\gtrsim  \min \inset{ \sigma \sqrt{ \frac{ r d } {|\Omega| \log(d) } }, \beta \sqrt{ \frac{ d }{ |\Omega| \log^3(d)  } }}. \]

If additionally we assume that $\bW$ is ``flat'' in the sense that the largest entry is no larger than a constant times the smallest entry, then we may conclude the stronger result that

\[ \frac{ \fronorm{ \bW^{(1/2)} \had \inparen{ \bM - \hat{\bM} } }}{\fronorm{ \bW^{(1/2)} }}  \gtrsim \min \inset{ \sigma \sqrt{ \frac{ r d }{|\Omega| }}, \frac{\beta}{ \sqrt{ \log(d) } } } \]
\end{theorem}

In particular, is $\sigma$ is on the order of $\beta$, the upper and lower bounds are approximately the same (up to logarithmic factors and factors of $r$).

Next, we state our results for approximately rank-$r$ matrices.

\begin{theorem}[Results for approxmately rank-$r$ matrices when $\Omega \sim \bW$, Informal]
Let $\bW \in \R^{d \times d}$ be a rank-1 matrix so that for all $i,j$, $1/d \leq W_{ij} \leq 1.$ 
Choose $\Omega \sim \bW$ as described above.

\textbf{Upper bound:}
With probability at least $1 - O(1/d)$ over the choice of $\Omega$, the following holds.
There is an algorithm $\cA$ so that for any $d \times d$ matrix $\bM \in \beta \sqrt{r} B_{\max}$,
$\mathcal{A}$ returns $\hat{\bM} = \cA( \bM_\Omega + \bZ_\Omega)$
so that with probability at least $1 - 1/d$ over the choice of $\bZ$,
\[
\frac{ \fronorm{ \bW^{(1/2)} \had ( \bM - \hat{ \bM } ) } }{ \fronorm{ \bW^{(1/2) } }} \lesssim
\beta \inparen{ \frac{ r^2 d }{|\Omega|} }^{1/4} \log^{1/2}(d)  +  \sqrt{\beta \sigma} \inparen{ \frac{ rd }{|\Omega| } }^{1/4} \log^{1/4}(d). \]

\textbf{Lower bound:} On the other hand, suppose additionally that $\bW$ is ``flat,'' in the sense that
the sense that the largest entry is no larger than a constant times the smallest entry.
Then with probability at least $1 - e^{-O( \fronorm{ \bW^{(1/2)} }^2) }$ over the choice of $\Omega$,
for any algorithm that only sees the values $\bM_\Omega + \bZ_\Omega$ and returns $\hat{M}$,
there is some $\bM \in \beta\sqrt{r} B_{\max}$ so that with probability at least $1/2$ over the choice of $\bZ$, 
\[ \frac{ \fronorm{ \bW^{(1/2)} \had \inparen{ \hat{\bM} - \bM } } } { \fronorm{ \bW^{(1/2)} } }  \gtrsim \sqrt{ \beta \sigma } \inparen{ \frac{r d }{ |\Omega| } }^{1/4}. \]
\end{theorem}

Again, this lower bound is tight up to logarithmic factors and factors of $r$ in the case that $\sigma \approx \beta$ and $\bW$ is reasonably ``flat.''

\subsection{Case study: When $\lambda$ is large}


Next we focus on the case when $\lambda$ is large. We assume the sampling pattern is symmetric so we may consider real eigenvalues. In order to prove lower bounds here, we make a few assumptions, in particular that the top \em two \em eigenvectors of $\ind{\Omega}$ are ``flat,'' in the sense that the largest element is no larger than a constant times the smallest.  

\begin{example}\label{ex:fatcycle}
Our running example is the following extreme sampling pattern:
\begin{equation}\label{eq:badgraph}
\ind{\Omega_t} := \begin{bmatrix} 
1 & 1 & 1 & 0 & \cdots & 0 & 0 & 1 & 1 \\
1 & 1 & 1 & 1 & \cdots & 0 & 0 & 0 & 1 \\
1 & 1 & 1 & 1 & \cdots & 0 & 0 & 0 & 0 \\
0 & 1 & 1 & 1 & \cdots & 0 & 0 & 0 & 0 \\
0 & 0 & 1 & 1 & \cdots & 0 & 0 & 0 & 0 \\
\vdots & & & & \ddots & & & & \vdots \\
0 & 0 & 0 & 0 & \cdots & 1 & 1 & 0 & 0 \\
0 & 0 & 0 & 0 & \cdots & 1 & 1 & 1 & 0 \\
0 & 0 & 0 & 0 & \cdots & 1 & 1 & 1 & 1 \\
1 & 0 & 0 & 0 & \cdots & 1 & 1 & 1 & 1 \\
1 & 1 & 0 & 0 & \cdots & 0 & 1 & 1 & 1 
\end{bmatrix} \in \{0,1\}^{d \times t}
\end{equation}
where there are $t$ ones per row and $t$ is odd.  That is, $\ind{\Omega_t}$ is the symmetric circulant matrix whose first row is
\[ \bz = ( \underbrace{ 1, 1, \ldots, 1 }_{(t+1)/2}, 0, 0, \ldots, 0, 0, \underbrace{1 , 1, \ldots, 1 }_{(t-1)/2} ). \]
It is well-known that the eigenvectors of $\ind{\Omega_t}$ are given by the rows of the discrete cosine transform (in particular, they satisfy the flatness condition above) and that the eigenvalues are given by the elements of $\mathbf{F} \bz$ where $\mathbf{F}$ is the discrete Fourier transform.  In particular, the largest eigenvalue of $\ind{\Omega_t}$ is $t$, and the second largest is 
\[ \lambda_2( \ind{ \Omega_t} ) = \sum_{\ell = -(t-1)/2 }^{ (t+1)/2 } \omega^\ell = \omega^{(1-t)/2} \inparen{ \frac{ \omega^{t+1} - 1 }{ \omega - 1 } }, \]
where
\[ \omega = e^{-2\pi i / d } \]
is a primitive $d$'th root of unity.
Now, we may compute
\[ \inabs{ \lambda_2( \ind{\Omega_t} ) } = \inabs{ \frac{ \omega^{t+1} - 1 }{ \omega - 1 } } = \sqrt{ \frac{ 1 - \cos( 2t\pi/ d ) }{ 1 - \cos( 2\pi / d ) } }, \]
which is at least
\[ \sqrt{ \frac{ 1 - \cos( 2t\pi/ d ) }{ 1 - \cos( 2\pi / d ) } } 
\geq t\inparen{ 1 - c (t/d)^2   } \]
for some constant $c$ (using the Taylor expansion for cosine to bound both terms).
In particular, it is quite close to $\lambda_1 = t$.  This means that even if we choose $\bW$ to be the rank-1 matrix that is as close as possible to $\ind{\Omega_t}$ (which in this case would be $\bW = \frac{t}{d} \mathbf{1} \mathbf{1}^T$), we still have
\[ \lambda = \opnorm{ \bW^{(1/2)} - \bW^{(-1/2)} \had \ind{\Omega} } = \opnorm{ \sqrt{ \frac{t}{d} } \mathbf{1} \mathbf{1}^T - \sqrt{ \frac{d}{t} } \ind{\Omega_t} } = \sqrt{ \frac{d}{t} } \lambda_2 = \Theta( \sqrt{ dt } ). \]
Thus, $\lambda$ is quite large.

\end{example}

Now, returning the the general case (provided that the top eigenvectors of $\ind{\Omega}$ are flat), 
suppose that we do choose $\bW$ to be the best rank-$1$ approximation to $\ind{\Omega}$.  This is a reasonable choice because it is an easy-to-compute matrix which intuitively makes $\lambda$ small.
Under these assumptions, it is not hard to work out what happens to the upper bound, which we do in Theorem~\ref{thm:gapupper}.  We are also able to prove a lower bound in Theorem~\ref{thm:gaplower}.  We informally record these results below.

\begin{theorem}[Bounds for rank-$k$ matrices so that $\ind{\Omega}$ is balanced and has a big spectral gap, informal]
Fix $\Omega \in [d] \times [d]$.  Suppose that $\bW$ is the best rank-1 approximation to $\ind{\Omega}$ and suppose that $\bW$ is flat in the sense that the largest entry is no larger than a constant times the smallest entry.  Suppose also that the second eigenvector of $\ind{\Omega}$ is flat in the same sense.  Suppose that the entries of $\bZ$ are i.i.d. $\gN(0,\sigma^2)$.

Suppose that $\bM \in \RR^{d \times d}$ has rank $r$ and $\infnorm{\bM} \leq \beta$, and let $\bY = \bM + \bZ$ where the entries of $\bZ$ are i.i.d. $\gN(0,\sigma^2)$.

\textbf{Upper bound:}  There is an algorithm $\cA$ so that for any rank-$r$ matrix $\bM$ with $\infnorm{\bM} \leq \beta$, $\cA$ returns $\hat{\bM} = \cA( \bM_\Omega + \bZ_\Omega )$ so that
 with probability at least $1 - 1/d$ over the choice of $\bZ$,
\[ \frac{\fronorm{ \bW^{(1/2)} \had ( \bM - \hat{ \bM } ) }}{ \fronorm{ \bW^{(1/2)} } } \lesssim  r\beta \inparen{ \frac{ \lambda_1 }{ \lambda_2 } } + \sigma \sqrt{ \frac{ r\log(d) }{ \lambda_1 } }. \] 

\textbf{Lower bound:} On the other hand, for any such algorithm $\cA$ that only sees the values $\bM_\Omega + \bZ_\Omega$ and returns $\hat{\bM}$, there is some rank $r$ matrix $\bM$ with $\infnorm{\bM} \leq \beta$ so that with probability at least $1/2$ over the choice of $\bZ$, 
\[ \frac{ \fronorm{ \bW^{(1/2)} \had ( \bM - \hat{ \bM} ) } }{ \fronorm{ \bW^{(1/2)} } } \gtrsim \min \inset{ \frac{ \beta }{ \sqrt{ r\log(d) } } , \sigma \sqrt{ \frac{ r }{ \lambda_1 - \lambda_2 } } }. \]
\end{theorem}

We note that the lower bound and the upper bound do not match.  However, the lower bound does capture \em some \em dependence on the gap between $\lambda_1$ and $\lambda_2$, and in particular the bounds match when this gap is very small.
In particular, if $\lambda_2 = \lambda_1 - O(1)$ and $\sigma \approx \beta$ then the upper bound essentially reads
\[ \frac{\fronorm{ \bW^{(1/2)} \had ( \bM - \hat{ \bM } ) }}{ \fronorm{ \bW^{(1/2)} } } \lesssim  r\beta, \]
which is trivial given that estimating $\hat{\bM} = \mathbf{0}$ will result in a weighted per-entry error bound of at most $\beta$.  However, the lower bound shows that in this case a non-trivial guarantee is impossible: it essentially reads
\[ \frac{ \fronorm{ \bW^{(1/2)} \had ( \bM - \hat{ \bM} ) } }{ \fronorm{ \bW^{(1/2)} } } \gtrsim \frac{ \beta }{ \sqrt{ r \log(d) } }. \]
Thus, in this extreme case, the upper and lower bounds match up to polynomial factors in $r$ and $\log(d)$.  

\begin{example}
Returning to our example of $\Omega_t$ above, we see that 
\[ \lambda_1 - \lambda_2 = O\left( \frac{ t^3 }{ d^2 } \right) \]
for some constant $t$.  This is $O(1)$ when $t = d^{2/3}$.  Thus we conclude from the analysis above that for this particular sampling pattern $\Omega_{d^{2/3}}$, one cannot recover even a rank-$1$ matrix $\bM$ in the presence of Gaussian noise significantly better than by just guessing $\hat{\bM} = \mathbf{0}$.  We note that in this case, the number of observations is $d^{5/3}$, which in the uniform sampling case would be more than enough to recover a rank-$1$ matrix.
\end{example}

\section{Related Work}\label{sec:related}

There are two lines of work that are related to ours.  The first is a line of work on \em deterministic \em or \em universal \em matrix completion.  In this line of work, one asks: what guarantees can one get for a sampling pattern $\Omega$ that are simultaneously valid on \em all \em matrices $\bM$.  (Notice that this question is interesting even if $\Omega$ is random to begin with: there is a big difference between the universal guarantee that ``with high probability, $\Omega$ is good for all $\bM$,'' and the randomized guarantee that ``for all $\bM$, $\Omega$ is good with high probability.'')

The second is a line of work where $\Omega$ is sampled randomly, but from biased distributions.  Our first case study (when $\Omega$ is drawn according to a weight matrix $\bH$) does give universal guarantees, but our results are also interesting from the perspective of sampling from biased distributions.

We briefly review both of these areas in more detail below.

\subsection{Deterministic/Universal matrix completion}

The works of Heiman, Schechtman, and Shraibman~\cite{heiman2014deterministic}, Bhojanapalli and Jain~\cite{Bhojanapalli2014} and Li, Liang, and Risteski~\cite{li2016recovery} relate the sampling pattern $\Omega$ to a graph whose adjacency matrix is given by $\ind{\Omega}$.  
Those works show that as long as this pattern is suitably close to an expander graph---in particular, if the deterministic sampling pattern is sufficiently uniform---then efficient recovery is possible, provided that the matrix $\bM$ is sufficiently incoherent.

There are also works which aim to understand when there is a unique (or only finitely many) low-rank matrices $\bM$ that can complete $\bM_\Omega$ as a function of the sampling pattern $\Omega$.  For example, \cite{PBN16} gives conditions on $\Omega$ under which there are only finitely many low-rank matrices that agree with $\bM_\Omega$, \cite{SXZ18} give a condition under which the matrix can be locally uniquely completed. The works \cite{AAW18} generalizes these results to the setting where there is sparse noise added to the matrix.  The works \cite{PN16, AWA17} study when rank estimation is possible as a function of a deterministic $\Omega$, and \cite{AW17} studies when a low-rank tensor can be uniquely completed.  Recently, \cite{Chat19} gave a combinatorial condition on $\Omega$ which characterizes when a low-rank matrix can be recovered up to a small error in the Frobenius norm from observations in $\Omega$ and showed that nuclear norm minimization will approximately recover $\bM$ whenever it is possible.

So far, all the works mentioned are interested in when recovery of the entire matrix (as measured by un-weighted Frobenius norm) is possible.  In our work, we are also interested in the case when such recovery is \em not \em possible.  To that end, we introduce a weighting matrix $\bH$ to capture this.  See \cite{Kiraly2015} for an interesting alternative, that uses an algebraic approach to answer when an entry can be completed or not.

The notion of weights has been studied before.
The work \cite{heiman2014deterministic} of Heiman et al. shows that for any weighting matrix $\bH$, there is a deterministic sampling pattern $\Omega$ and an algorithm that observes $\bM_\Omega$ and returns $\hat{\bM}$  so that $\fronorm{ \bH \had ( \bM - \bM ) }$ is small.  
Their algorithm can be informally described as finding the matrix with the smallest $\gamma_2$-norm that is correct on the observed entries.
Lee and Shraibman~\cite{lee2013matrix} start with this framework, 
and study the more general class of algorithms that find the ``simplest'' matrix that is correct on the observed entries, where ``simplest'' can mean of smallest norm, smallest rank, or a broad class of definitions.  (We note that this algorithm is not efficient in general).
That work gives a way of measuring which deterministic sampling patterns $\Omega$ are good with respect to a weight matrix $\bH$. 
Similar to our work, they introduce a parameter which measures the complexity of the sampling pattern $\Omega$.  Unlike our work, their parameter is a bit more complicated, and is obtained by solving a semidefinite program involving $\Omega$.  They show that if this parameter is small, then the entries of the matrix can be recovered with appropriate weights.  Moreover, they show, given $\Omega$, how to efficiently compute a weight matrix $\bH$ so that the performance of the algorithm is optimal. 

To summarize, there are a few main differences between our work and previous work:
\begin{itemize}
	\item We are interested in cases where it may not be easy to estimate $\bM$ in Frobenius norm from the noisy samples $\bM_\Omega + \bZ_\Omega$, which was the goal in~\cite{Bhojanapalli2014,li2016recovery}.  (And certainly we may not be able to uniquely recover $\bM$, as is the goal in \cite{PN16, AWA17}).
	\item We are interested in efficient algorithms, while the algorithms of \cite{lee2013matrix} need not be efficient. 
	 In particular, we focus on debiased projection-based algorithms.  These algorithms are extremely simple (computationally and intuitively) and are provably optimal in some cases when there is noise or when $\bM$ need not be exactly low-rank.  However, in the non-noisy exactly low-rank case, our algorithm need not recover the matrix exactly; the algorithm of \cite{Bhojanapalli2014} is able to do this.
	\item We are interested in \em rank-1 \em weighting patterns.  This is more restrictive than the works \cite{lee2013matrix, heiman2014deterministic}, but allows us both to obtain efficient algorithms and to prove lower bounds.
\end{itemize}

\subsection{Weighted matrix completion and matrix completion from biased samples}
Weighted matrix completion has appeared in several works (see e.g. \cite{heiman2014deterministic,lee2013matrix,negahban2010restricted,eftekhari2016weighted}) under the assumption of random biased sampling.  The connection between weighting and biased sampling is most easily expressed in the supervised learning setup.  Indeed, let $\bD \in \R^{d \times d}$ encode a random distribution over matrix indices so that
\[0 \leq D_{i,j} \leq 1, \qquad \text{and} \qquad \sum_{i,j} D_{i,j} = 1.\]
Let one observation take the form $Y_{i,j} = M_{i,j} + Z_{i,j}$ where $(i,j)$ is sampled randomly according to $\bD$ and suppose you are given $m$ independent observations of this form (allowing repetition of matrix entries).  Considering squared \textit{loss function}, the \textit{excess risk} of an estimator $\hat{\bM}$ is
\[\E{(\hat{M}_{i,j} - Y_{i,j})^2 - (M_{i,j} - Y_{i,j})^2} = \sum_{i,j} D_{i,j}(\hat{M}_{i,j} - M_{i,j})^2 = \fronorm{\bD^{1/2} \had (\hat{\bM} - \bM)}^2.\]
Bounding the excess risk then gives a weighted error bound for the estimator.  

In \cite{negahban2010restricted}, the authors consider the case when rank$(\bD)$ = 1, which almost corresponds with the random model given in our Section \ref{sec: small lambda}.  They identify a certain \textit{restricted eigenvalue condition} which holds with high probability under the random model.  In the random model, there error bounds are similar to the ones in our paper.  However, the \textit{restricted eigenvalue condition} is not known to be verifiable in polynomial time.  They also give a lower bound, essentially matching ours, but under the assumption of uniformly random sampling.

Several other papers consider random sampling without making the assumption that rank($\bD$) = 1.  In \cite{klopp2012noisy}, the authors consider nuclear-norm minimization in the case when the sampling distribution is not uniform.  They give unweighted error bounds which degrade (as they should) as the sampling distribution becomes less uniform.  In \cite{cai2016matrix}, the authors allow general sampling distribution and consider the least squares estimator under a max-norm constraint.  By bounding the \textit{Rademacher complexity} of the max-norm ball, they bound the excess risk.  They are also able to extend this to the \textit{binary} setup \cite{cai2013max}.  Much of this analysis was based on previous works analyzing the max norm \cite{srebro2005rank, foygel2011concentration}.  

Although our focus is not on random sampling patterns, as we discuss below in Section \ref{sec: small lambda}, our analysis does imply some interesting consequences for that setting as well. In particular, one may consider the random model $\Omega \sim \bW$ by which we mean that the pattern $\Omega$ is obtained by sampling the (i,j)th entry with probability $W_{ij}$. In \cite{negahban2010restricted}, the authors consider this random model for a rank-1 matrix $\bW$ and show that when $\bW \had \bM$ is (nearly) low rank and not too spiky, the solution $\hat{\bM}$ of a semidefinite program (SDP) yields a small error $\fronorm{ \bW^{(1/2)} \had ( \bM - \hat{\bM} ) }$.  
They also provide a theoretical lower bound for this setting, although unlike our results below, theirs becomes trivial when the spectral gap is large. In our work however, we focus on efficient projection algorithms for recovery rather than SDPs. In addition, we are able to provide uniform results in the sense that they hold with high probability for all matrices. 

Other works that consider the random model $\Omega \sim \bW$ include \cite{srebro2010collaborative}, where in fact like our work, the authors also consider the loss function $\fronorm{ \bW^{(1/2)} \had (\bM - \hat{\bM} ) }$, but without assuming $\bW$ is low-rank. There, the authors consider random but non-uniform sampling pattern distributions (rather than fixed patterns as in our work), and establish lower bounds for this model using trace-norm minimization. In fact, \cite{srebro2005generalization} shows that the trace norm is a good proxy for the rank in terms of sampling guarantees in the sense that if there does happen to be a matrix $\bM*$ that is close to $\bM$ in this weighted Frobenius norm, then (computationally intractable) minimization with rank constraints can recover $\hat{\bM}$ approximately as close to $\bM$ as $\bM*$ is, using on the order of $rn$ samples via $\Omega \sim \bW$.  In some sense, our results can be viewed as a generalization to those of \cite{srebro2010collaborative}, in that our lower bounds hold for any algorithm and the random model can be viewed as generating a special case of our result. 

Our results can also be viewed as generalizations of those that use alternative sampling strategies adapted for e.g. coherence matrices. In \cite{CBSW15},  the authors show that by sampling according to \textit{leverage scores}, one can recover coherent matrices using nuclear norm minimization.  Our setting can be cast in this framework, which discuss in detail in Section~\ref{sec:nogapapp}. However, this again isn't our main focus, as we are focusing on efficient algorithms and establishing universal and uniform lower bounds. Nonetheless, one can construct appropriate weight matrices $\bW$ that yield sampling distributions related to the leverage scores.  See Section \ref{sec:nogapapp} for comparison and more discussion.

Other works that incorporate non-uniform sampling include \cite{meka2009matrix}, which proposes a graph-theoretic algorithm for matrix completion when the entries are power-law distributed. Lastly, 
\cite{liu2018samples} proposes a so-called \textit{isomeric} property for sampling patterns, viewed as a generalization of low-rankness, that guarantees (exact) matrix completion. There, the authors show that uniform sampling implies the isomeric condition, but this condition is a weaker assumption than uniformity, and propose a  Schatten quasi-norm induced method for recovery. We again take a different approach than these works, focusing on simple recovery and universal bounds.

\section{General Upper Bounds}\label{sec:upper}
In this section we prove general upper bounds for weighted recovery of low-rank, or approximately low-rank, matrices from deterministic sampling patterns.

\subsection{Bounds for rank-$r$ matrices}

\begin{theorem}[General upper bound for rank-$r$ matrices]\label{lem:upperrankk}
Let $\bW = \bw \bu^T \in \RR^{d_1 \times d_2}$ have strictly positive entries, and fix $\Omega \subseteq [d_1] \times [d_2]$.  Suppose that $\bM \in \RR^{d_1 \times d_2}$ has rank $r$.  Suppose that $Z_{ij} \sim \gN(0,\sigma^2)$ and let 
\[ \hat{ \bM } = \bW^{(-1/2)} \had \argmin_{ \text{rank}(\bX) = r } \norm{ \bX - \bW^{(-1/2)} \had (\bM_\Omega + \bZ_\Omega) }. \]
Then with probability at least $1 - 1/(d_1 + d_2)$ over the choice of $\bZ$,
\[ \fronorm{ \bW^{(1/2)} \had ( \bM - \hat{ \bM } ) } \leq 2\sqrt{2}r \lambda \infnorm{ \bM } +  4\sqrt{2} \sigma \mu \sqrt{r \log(d_1 + d_2 )}, \]
where $\lambda$ and $\mu$ are as in \eqref{eq:lambda} and \eqref{eq:mu}, respectively.
\end{theorem}

\begin{proof}
Let $\bY = \bM + \bZ$.
Observe that $\bM, \hat{\bM}$ are both rank $r$ and hence $\bW^{(1/2)} \had ( \bM - \hat{\bM} )$ is at most rank $2r$.  Thus,
\begin{align*}
\fronorm{ \bW^{(1/2)} \had ( \bM - \hat{ \bM } ) } &\leq \sqrt{2r} \opnorm{ \bW^{(1/2)} \had (\bM - \hat{\bM}) } \\
&\leq \sqrt{2r} \inparen{ \opnorm{ \bW^{(1/2)} \had \hat{\bM} - \bW^{(-1/2)} \had \bY_\Omega }  + \opnorm{ \bW^{(1/2)} \had \bM - \bW^{(-1/2)} \had \bY_\Omega } } \\
&\leq 2\sqrt{2r} \opnorm{ \bW^{(1/2)} \had \bM - \bW^{(-1/2)} \had \bY_\Omega }, 
\end{align*}
using the definition of $\hat{\bM}$ in the final line.  Then we bound 
\begin{align*}
\opnorm{ \bW^{(1/2)} \had \bM - \bW^{(-1/2)} \had \bY_\Omega }
&\leq \opnorm{ \bW^{(1/2)} \had \bM - \bW^{(-1/2)} \had \bM_\Omega } + \opnorm{ \bW^{(-1/2)} \had \bZ_\Omega } \\
&= \opnorm{ \bM \had \inparen{ \bW^{(1/2)} - \bW^{(-1/2)} \had \ind{\Omega} } } + \opnorm{ \bW^{(-1/2)} \had \bZ_\Omega } \\
&\leq \maxnorm{ \bM } \cdot \lambda +  \opnorm{ \bW^{(-1/2)} \had \bZ_\Omega } \\
&\leq \sqrt{r} \infnorm{\bM} \lambda +  \opnorm{ \bW^{(-1/2)} \had \bZ_\Omega },
\end{align*}
using the fact that $\bM$ is rank $r$ and hence $\maxnorm{ \bM } \leq \sqrt{r} \infnorm{ \bM }$.
Thus we conclude that
\begin{equation}\label{eq:errorbound}
\fronorm{ \bW^{(1/2)} \had ( \bM - \hat{ \bM } ) } \leq 2\sqrt{2}r \lambda \infnorm{ \bM } +  2\sqrt{2r} \opnorm{ \bW^{(-1/2)} \had \bZ_\Omega },  
\end{equation}
and it remains to bound the second term.
We have
\[ \bW^{(-1/2)} \had \bZ_\Omega = \sum_{i=1}^{d_1} \sum_{j=1}^{d_2} \frac{ \ind{ (i,j) \in \Omega } Z_{ij} }{ \sqrt{ W_{ij} } } \be_i \be_j^T, \]
where $\be_i$ is the $i$'th standard basis vector.  We may apply Theorem~\ref{thm:matrixgaussian} with 
\[ \bX_{ij} = \frac{ \ind{ (i,j) \in \Omega } }{ \sqrt{ W_{ij}}} \be_i \be_j^T. \]
We have
\[ \opnorm{\sum_{i,j} \bX_{ij} \bX_{ij}^T} = \opnorm{ \sum_i \inparen{ \sum_j \frac{ \ind{(i,j) \in \Omega }}{ W_{ij} } } \be_i \be_i^T } = \max_i \sum_j \frac{\ind{(i,j) \in \Omega }}{W_{ij}} \leq \mu^2 \]
and similarly
\[ \opnorm{ \sum_{i,j} \bX_{ij}^T\bX_{ij} } = \max_j \sum_i \frac{ \ind{(i,j) \in \Omega } }{ W_{ij} } \leq \mu^2 \]
and by Theorem~\ref{thm:matrixgaussian},  for any $t > 0$,
\[ \PR{ \opnorm{ \bW^{(-1/2)} \had \bZ_\Omega } \geq t } \leq 2(d_1 + d_2) \exp\inparen{ \frac{ - t^2 }{ 2 \sigma^2 \mu^2 } }. \]
We conclude that with probability at least $1 - \frac{1}{d_1 + d_2}$, we have
\[ \opnorm{ \bW^{(-1/2)} \had \bZ_\Omega } \leq 2 \sigma \mu \sqrt{ \log(d_1 + d_2) }. \]
Plugging this into \eqref{eq:errorbound} proves the theorem.
\end{proof}


\subsection{Bounds for approximately rank-$r$ matrices}
In this section we prove a bound analogous to Theorem~\ref{lem:upperrankk} for the case when $\bM \in \beta\sqrt{r} B_{\max}$ is only approximately low rank.  We use the same simple projection algorithm, except this time we project onto the max norm ball instead of onto the cone of rank $r$ matrices.

\begin{theorem}[General upper bound for approximately rank-$r$ matrices]\label{lem:uppermaxball}
There is a constant $C$ so that the following holds.
Let $\bW = \bw \bu^T \in \RR^{d_1 \times d_2}$ have strictly positive entries, and fix $\Omega \subseteq [d_1] \times [d_2]$.  Suppose that $\bM \in \RR^{d_1 \times d_2}$ has $\maxnorm{ \bM } \leq \beta\sqrt{r}$.  Suppose that $Z_{ij} \sim \gN(0,\sigma^2)$ and let 
\[ \hat{ \bM } = \bW^{(-1/2)} \had \argmin_{ \maxnorm{ \bX } \leq \beta \sqrt{r} } \norm{ \bX - \bW^{(-1/2)} \had (\bM_\Omega + \bZ_\Omega) }. \]
Then with probability at least $1 - 1/(d_1 + d_2)$ over the choice of $\bZ$,
\[ \fronorm{ \bW^{(1/2)} \had ( \bM - \hat{ \bM } ) } \leq C \cdot \fronorm{ \bW^{(1/2)} }^{1/2} \inparen{ \beta \sqrt{ r\lambda } + \sqrt{ \beta \sigma } \inparen{ \mu^2 r \log(d_1+d_2) }^{1/4}} \] 
where
where $\lambda$ and $\mu$ are as in \eqref{eq:lambda} and \eqref{eq:mu}, respectively.
\end{theorem}

\begin{proof}Let $\bY = \bM + \bZ$, and let $\bQ = \bM - \hat{\bM}$.  Then we have
\begin{equation}\label{eq:firststep}
\fronorm{ \bW^{(1/2)} \had (\bM - \hat{\bM} ) }^2 = \ip{ \bW \had \bQ }{ \bQ } \leq \maxnorm{ \bQ } \| \bW \had \bQ \|_{ \max^* }. 
\end{equation}
The first factor we can bound by
\[ \maxnorm{ \bQ } \leq 2\beta \sqrt{r}, \]
by the assumption on $\bM$ and the definition of $\hat{\bM}$.  For the second factor, we have
\begin{align*}
\| \bW \had \bQ \|_{\max^*} &\leq K_G \max_{ \ba \in \{\pm 1 \}^{d_1}, \bb \in \{\pm 1 \}^{d_2} } \ip{ \ba \bb^T }{ \bW \had \bQ } \\
&= K_G \max_{ \ba \in \{\pm 1 \}^{d_1}, \bb \in \{\pm 1 \}^{d_2} } \ip{ \ba \bb^T \had \bW^{(1/2)}}{ \bW^{(1/2)} \had \bQ } \\
&\leq K_G \max_{ \ba \in \{\pm 1 \}^{d_1}, \bb \in \{\pm 1 \}^{d_2} } \nucnorm{ \ba \bb^T \had \bW^{(1/2)} } \opnorm{ \bW^{(1/2)} \had \bQ } \\
&= K_G \nucnorm{ \bW^{(1/2)} } \opnorm{ \bW^{(1/2)} \had \bQ } \\
&= K_G \fronorm{ \bW^{(1/2)} } \opnorm{ \bW^{(1/2)} \had \bQ }
\end{align*}
where in the last line we have used the fact that $\bW^{(1/2)}$ is rank $1$ and so the nuclear norm is equal to the Frobenius norm.
Then we bound
\begin{align*}
\opnorm{ \bW^{(1/2)} \had \bQ } &= \opnorm{ \bW^{(1/2)} \had ( \bM - \hat{\bM } ) } \\
&\leq \opnorm{ \bW^{(1/2)} \had \bM - \bW^{(-1/2)} \had \bY_\Omega } + \opnorm{ \bW^{(1/2)} \had \hat{\bM} - \bW^{(-1/2)} \had \bY_\Omega } \\
&\leq 2 \opnorm{ \bW^{(1/2)} \had \bM - \bW^{(-1/2)} \had \bY_\Omega } \\
&\leq 2 \inparen{ \opnorm{ \inparen{ \bW^{(1/2)} - \bW^{(-1/2)} \had \ind{\Omega} } \had \bM } + \opnorm{ \bW^{(-1/2)} \had \bZ_\Omega } } \\
&\leq 2 \inparen{ \maxnorm{ \bM } \lambda + 2 \sigma \mu \sqrt{ \log(d_1 + d_2 ) }},
\end{align*}
using in the last line the analysis from the proof of Theorem~\ref{lem:upperrankk}.  The putting it together with \eqref{eq:firststep}, we have 
\begin{align*}
\fronorm{ \bW^{(1/2)} \had (\bM - \hat{\bM} ) }^2 &\leq \maxnorm{ \bQ } \| \bW \had \bQ \|_{ \max^* } \\
&\leq 2 \beta \sqrt{r} \inparen{ K_G \fronorm{ \bW^{(1/2)} } 2 \inparen{ \beta \sqrt{r} \lambda + 2 \sigma \mu \sqrt{ \log(d_1 + d_2) } } }. 
\end{align*}
Taking the square root and choosing $C$ appropriately completes the proof.
\end{proof}

\section{General Lower Bounds}\label{sec:lower}
As we will see in our case studies in Sections~\ref{sec:nogap} and \ref{sec:gap}, the upper bounds from Section~\ref{sec:upper} are tight in some situations.  In order to prove lower bounds in those specific settings, in this section we give general lower bounds which can be specialized to both the exactly rank-$r$ and the approximately-rank-$r$ settings.
Our lower bounds all rest on Fano's Inequality, which we recall below. 
\begin{theorem}[Fano's Inequality]\label{thm:fano}
Let $\mathcal{F} = \inset{ f_0, \ldots, f_n }$ be a collection of densities on $\mathcal{X}$, and suppose that $\mathcal{A}: \mathcal{X} \to \{0, \ldots, n\}$.  Suppose there is some $\beta > 0$ so that for any $i \neq j$, $\KL{f_i}{f_j} \leq \beta$.  Then 
\[ \max_i \mathbb{P}_{X \sim f_i} \inset{ \mathcal{A}(X) \neq i } \geq 1 - \frac{ \beta + \log(2) } {\log(n)}. \]
\end{theorem}

The following lemma specializes Fano's inequality to our setting.
\begin{lemma}\label{lem:fano}
Let $K \subset \RR^{d_1 \times d_2}$, and let $\mathcal{X} \subset K$ be a finite subset of $K$ so that $|\cX| > 16$.  Let $\Omega \subseteq [d_1] \times [d_2]$ be a sampling pattern.  
Let $\sigma > 0$ and 
choose
\[ \kappa \leq \frac{ \sigma \sqrt{ \log|\cX| } }{ 4 \max_{ \bX \in \cX } \fronorm{ \bX_\Omega }}, \]
and suppose that 
\[ \kappa \cX \subseteq K. \]
Let $\bZ \in \mathbb{R}^{d_1 \times d_2}$ be a matrix whose entries $Z_{i,j}$ are i.i.d., $Z_{i,j} \sim \gN(0, \sigma^2)$.  Let $\mathbf{H} \subseteq \mathbb{R}^{d_1 \times d_2}$ be any weight matrix.  

Then for any algorithm $\mathcal{A}: \mathbb{R}^\Omega \to \mathbb{R}^{d_1 \times d_2}$ that
takes as input $\bX_\Omega + \bZ_\Omega$ for $\bX \in K$ and outputs an estimate $\widehat{\bX}$ to $\bX$, there is some $\bM \in K$ so that 
\[ \fronorm{ \bH \had ( \cA( \bM_\Omega + \bZ_\Omega ) - \bM } \geq 
\frac{ \kappa }{2} \min_{ \bX \neq \bX' \in \cX } \fronorm{ \bH \had ( \bX - \bX' ) } \]
with probability at least 1/2.
\end{lemma}

\begin{proof}
Consider the net
\[ \cX' = \inset{ \kappa \bX \suchthat \bX \in \cX } \]
which is a scaled version of $\cX$.  By assumption, $\cX' \subseteq K$.

Recall that the KL divergence between two multivariate Gaussians is given by
\[\KL{ \gN(\mu_1, \Sigma_1) }{\gN(\mu_2, \Sigma_2 ) } = \frac{1}{2} \inparen{ \log\frac{ \det\Sigma_2}{\det \Sigma_1 } - n + \tr( \Sigma_2^{-1} \Sigma_1 ) + \ip{\Sigma_2^{-1} (\mu_2 - \mu_1)}{ \mu_2 - \mu_1 } }. \]
Specializing to $\bU, \bV \in \cX'$, with $\bI = \bI_{\Omega \times \Omega}$, 
\begin{align*}
 \KL{ \bU_\Omega + \bZ_\Omega }{  \bV_\Omega + \bZ_\Omega } &= \KL{ \gN(\bU_\Omega, \sigma^2 \bI)}{\gN(\bV_\Omega, \sigma^2 \bI) } \\
&= \frac{ \fronorm{ \bU_\Omega - \bV_\Omega }^2 }{ 2\sigma^2 }  \\
&\leq \max_{ \bX' \in \cX' } \frac{ \fronorm{ \bX' } }{ \sigma^2 } \\
& = \frac{ \kappa^2 \max_{ \bX \in \cX } \fronorm{ \bX }^2 }{ \sigma^2 }.
\end{align*}

Suppose that $\mathcal{A}$ is as in the statement of the lemma.
Define an algorithm $\overline{\cA}:\RR^{\Omega} \to \RR^{d_1 \times d_2}$ so that $\overline{\cA}(\bY) = \bX$ for the unique $\bX \in \cX'$ so that 
\[ \fronorm{ \bH \had (\bX - \cA( \bY )) } < \frac{1}{2} \min_{ \bX \neq \bX' \in \cX' } \fronorm{ \bH \had ( \bX - \bX' ) } := \rho/2 \]
 if it exists, and $\overline{ \cA}(\bY) = \cA(\bY)$ otherwise. 

Then by Fano's inequality (Theorem \ref{thm:fano}), there is some $\bM \in \cX'$ so that
\begin{align*}
\PR{ \overline{\cA}( \bM_\Omega + \bZ_\Omega ) \neq \bM } 
& \geq 1 - \frac{ \max_{ \bX \in \cX' }\fronorm{ \bX_\Omega }^2 }{ \sigma^2 \log(|\cX| - 1 )} - \frac{ \log(2) }{ \log( |\cX| - 1 ) } \\
& = 1 - \frac{ \kappa^2 \max_{ \bX \in \cX }  \fronorm{ \bX_\Omega }^2 } { \sigma^2 \log( |\cX| - 1 ) } - \frac{ \log(2) }{ \log( |\cX| - 1 ) } \\
& \geq 1 - \frac{1}{4} - \frac{\log(2)}{ \log(|\cX| - 1) } \\
& \geq 1/2,
\end{align*}
using the assumption that $|\cX| \geq 16$ as well as the fact that
\[ \kappa \leq \frac{ \sigma \sqrt{ \log|\cX|} }{ 4 \max_{ \bX \in \cX } \fronorm{ \bX_\Omega } } \leq \frac{ \sigma \sqrt{ \log( |\cX| - 1 )} }{ 2 \max_{ \bX \in \cX } \fronorm{ \bX_\Omega } }. \]

If $\overline{\cA}(\bM_\Omega + \bZ_\Omega ) \neq \bM$, then $\fronorm{ \bH \had \cA( \bM_\Omega + \bZ_\Omega ) } > \rho/2$, and so
\[ \PR{ \fronorm{ \bH \had \cA( \bM_\Omega + \bZ_\Omega ) - \bM } \geq \rho/2 } \geq \PR{ \overline{ \cA} (\bM_\Omega + \bZ_\Omega) \neq \bM} \geq 1/2. \] 

Finally, we observe that
\[  
\frac{\rho}{2} = \frac{1}{2} \min_{ \bX \neq \bX' \in \cX'} \fronorm{ \bH \had ( \bX - \bX' ) }  
= \frac{\kappa}{2} \min_{ \bX \neq \bX' \in \cX } \fronorm{ \bH \had ( \bX - \bX' ) }, \]
which completes the proof.
\end{proof}

Our lower bounds in Sections~\ref{sec:nogap} and \ref{sec:gap} will follow from Lemma~\ref{lem:fano} by choosing an appropriate net.  Below we prove a general lemma about picking a net, which we will use multiple times in subsequent proofs.

\begin{lemma}\label{lem:net}
There is some constant $c$ so that the following holds.  Let $r, d_1, d_2 > 0$ be sufficiently large, and suppose that $d_1 \geq d_2$.  
Let $K$ be the cone of rank-$r$ matrices. Let $\bH$ be any rank-1 weight matrix, and let $\bA$ be any rank-1 matrix with $\|\bA\|_\infty \leq 1$.
Write $\bH = \bh \bg^T$ and $\bA = \ba \bb^T$, and let
\[ \bz = (\bg \had \bb)^{(2)} \qquad \bv = ( \bh \had \ba )^{(2)}. \]
Let 
\[ \gamma = c \sqrt{ r \log(d_1d_2) } \]
There is a net $\cX \subseteq K \cap \gamma B_\infty \cap r B_{\max}$ so that:
\begin{enumerate}
\item The net has size $|\cX| \geq N$, for

\[ N = 2e \exp \inparen{  c \cdot \min \inset{ \frac{ \onenorm{\bv}^2\onenorm{\bz}^2 }{ \twonorm{\bv}^2 \twonorm{\bz}^2 } , \frac{ \onenorm{ \bv } \onenorm{ \bz } }{ \infnorm{ \bv}  \twonorm{ \bz} \sqrt{ r \log(r) }} , \frac{ \onenorm{ \bv } \onenorm{ \bz }}{  \infnorm{ \bv} \infnorm{ \bz } r \log(r) }  , \frac{ r^2 \onenorm{ \bv}^2 }{ r \twonorm{ \bv}^2 } }}.
\]

\item $\fronorm{ \bX_\Omega } \leq \sqrt{ c\cdot r } \fronorm{ \bA_\Omega }$ for all $\bX \in \cX$.
\item $\fronorm{ \bH \had ( \bX - \bX' ) } \geq \sqrt{r} \fronorm{ \bA \had \bH }$ for all $\bX \neq \bX' \in \cX$.
\end{enumerate}
\end{lemma}

\begin{remark}
We do not need the assumption that $d_1 \geq d_2$ for the statement of Lemma~\ref{lem:fano} to be true; however, the result is stronger if $d_1 \geq d_2$, because in the cases we consider below (where $\bA$, $\bH$ are ``flat enough''), then the term in the minimum is dominated by $\frac{ r^2 \onenorm{ \bv }^2 }{ r \twonorm{ \bv }^2 } \approx r d_1$.  If $d_2 \geq d_1$, then we may switch the roles of $d_1$ and $d_2$ in the proof below and obtain a bound that depends on $d_2$.
\end{remark}

\begin{proof}
Let $\mathcal{L} \subset \{\pm 1\}^{d_1 \times r}$ be a set of random $\pm 1$-valued matrices chosen uniformly at random with replacement, of size $4N$.

Choose $\bR \in \{ \pm 1 \}^{d_2 \times r}$ to be determined below.
Let $\bl_i$, for $i \in [d_1]$, denote the rows of $\bL$, and similarly let $\br_i$ for $i \in [d_2]$ denote the rows of $\bR$.
Let
\[ \mathcal{X} = \inset{ \bA \had \bL \bR^T \suchthat \bL \in \mathcal{L} }. \]

(We note that if one wishes to prove a similar lemma for $d_2 > d_1$, then we should
make the net by choosing $\bR$ at random and fixing $\bL$.)

We begin by estimating the first requirement on $\fronorm{ \bX_\Omega }$, and also the requirement that $\infnorm{ \bX } \leq \gamma$ and $\|\bX\|_{\max} \leq r$ for all $\bX \in \cX$.  We have
\[ \EE \fronorm{ \bX_\Omega }^2 = \EE \sum_{i,j \in \Omega} A_{ij} \ip{\bl_i}{\br_j}^2 = r \fronorm{ \bA_\Omega }^2, \]
where the expectation is over the random choice of $\bL$.
By Markov's inequality, $\PR{ \fronorm{ \bX_\Omega }^2 > 4r\fronorm{ \bA_\Omega }^2 } \leq 1/4$.
We also have
\[ \infnorm{ \bX } = \max_{ i,j \in [d_1] \times [d_2] } |A_{ij}| |\ip{\bl_i}{\br_j}|. \]
Now, for each $i,j$, $\ip{\bl_i}{\br_j}$ satisfies
\[ \PR{ |\ip{ \bl_i }{ \br_j }| \geq t } \leq \exp \inparen{ \frac{ - 2t^2 }{ r } } \]
by Hoeffding's inequality. Using the fact that $|A_{ij}| \leq 1$ by assumption and a union bound over all $d_1d_2$ values of $i,j$, we conclude that 
\[ \PR{ \infnorm{ \bX } >  \sqrt{ r \log(4d_1d_2) / 2 } } \leq 1/4. \] 
Finally, by definition the matrices $\bX \in \cX$ satisfy $\maxnorm{ \bX } \leq r$, by writing
\[ \bX = (\bD_{\ba} \bL ) ( \bD_\bb \bR )^T \]
and observing that each row of $\bD_{\ba} \bL$ has $\ell_2$ norm at most $\|\ba\|_\infty \sqrt{r} \leq \sqrt{r}$ and similarly for each row of $\bD_{\ba} \bR$.

By a union bound, for one matrix $\bX \in \cX$, the probability that all of $\maxnorm{ \bX } \leq r$,  $\infnorm{ \bX } \leq \sqrt{ r \log( 4d_1d_2 )/2 }$ and $\fronorm{ \bX_\Omega }^2 \leq 4r \fronorm{ \bA_\Omega }^2 $ is at most $1/2$.
Thus, by a Chernoff bound it follows that with high probability, at least $1 - \exp( -CN)$ for some constant $C$, there are at least $|\cX|/4$ matrices $\bX \in \cX$ so that all of these hold.  Let $\tilde{\cX} \subset \cX$ be the set of such $\bX$'s.  The net guaranteed in the statement of the theorem will be $\tilde{\bX}$, which in the favorable case satisfies both items 1 and 2 in the lemma, and also is contained in $K \cap \gamma B_\infty$.

Thus, we turn our attention to item 3: we will show that this holds for $\cX$ with high probability, and so in particular it will hold for $\tilde{\cX}$, and this will complete the proof of the lemma.

Fix $\bX \neq \bX' \in \cX$, and write
\begin{align*}
\fronorm{ \bH \had ( \bX - \bX' ) }^2 &= \fronorm{ \bH \had \bA \had ( \bL - \bL' ) \bR }^2 \\
&= \sum_{i,j \in [d_1] \times [d_2]} H_{ij}^2 A_{ij}^2 \ip{ \bl_i - \bl_i' }{ \br_j }^2 \\
&= 4 \sum_{i,j \in [d_1] \times [d_2]} H_{ij}^2 A_{ij}^2 \ip { \bxi_i }{ \br_j }^2
\end{align*}
where we define $\bxi_i = \frac{1}{2}(\bl_i - \bl_i')$.  Thus, each entry of $\bxi_i$ is independently $0$ with probability $1/2$ or $\pm 1$ with probability $1/4$ each.  Rearranging the terms and recalling the definitions of $\bv$ and $\bz$ above, we have
\begin{equation}\label{eq:matrixform}
\fronorm{ \bH \had ( \bX - \bX' ) }^2 = 4 \sum_{i=1}^{d_1} v_i \bxi_i^T \bR^T \bD_{\bz} \bR \bxi_i, 
\end{equation}
where $\bD_{\bz}$ denotes the $d_2 \times d_2$ diagonal matrix with $\bz$ on the diagonal.

In order to understand \eqref{eq:matrixform}, we need to understand the matrix $\bR^T \bD_{\bz} \bR \in \R^{r \times r}$.  The diagonal of this matrix is $\onenorm{\bz} \bI$.  We will choose the matrix $\bR$ so that the off-diagonal terms are small.  More precisely, we will choose $\bR$ according to the following claim.
\begin{claim}
There is a matrix $\bR \in \{\pm 1\}^{d_2 \times r}$ so that:
\begin{enumerate}
\item[(a)]
$ \fronorm{ \bR^T \bD_{\bz} \bR - \onenorm{\bz} \bI }^2 \leq 2 r^2 \twonorm{ \bz }^2 $ and
\item[(b)]
$ \opnorm{ \bR^T \bD_{\bz} \bR - \onenorm{\bz} \bI } \leq 2 \inparen{ \twonorm{ \bz } \sqrt{ r\log(r)} + \infnorm{ \bz } r \log(r) }.$
\end{enumerate}
\end{claim}
\begin{proof}
Choose $\bR$ at random.  We will show that both (a) and (b) above happen with probability strictly greater than $1/2$, so by a union bound there exists a choice for $\bR$ which satisfies both.

First, for (a), we compute
\begin{align*}
\EE \fronorm{ \bR^T \bD_{\bz} \bR - \onenorm{\bz} \bI }^2 &= \sum_{i \neq j} \EE ( \be_i \bR^T \bD_z \bR \be_j )^2 \\
&= r (r-1) \twonorm{z}^2,
\end{align*}
which implies by Markov's inequality that
\[ \PR{ \fronorm{ \bR^T \bD_{\bz} \bR - \onenorm{\bz} \bI }^2 > 2r^2 \twonorm{ \bz }^2 } < \frac{1}{2}. \]
For (b), we write
\[  \bR^T \bD_{\bz} \bR - \onenorm{\bz} \bI = \sum_{i=1}^{d_2} z_i ( \br_i \br_i^T - \bI ), \]
which is a sum of mean-zero independent random matrices, so we apply the matrix Bernstein Inequality (Theorem \ref{thm:matrixbernstein}).  We have
for any $t > 0$,
\[ \PR{ \opnorm{ \sum_i z_i ( \br_i \br_i^T - \bI ) } > t } \leq r \exp \inparen{ \frac{ - t^2/2 }{ r\twonorm{\bz}^2 + rt \infnorm{ \bz}/3 } }, \]
using the fact that $\opnorm{ z_i (\br_i \br_i^T - \bI) } \leq \infnorm{\bz}(r-1) \leq \infnorm{ \bz}r$ for all $i$, and that
\[ \opnorm{ \EE \sum_i z_i^2 (\br_i \br_i^T - \bI )^2 } = \opnorm{ \sum_i z_i^2 (r-1)\bI } \leq r\twonorm{\bz}^2. \]
Choosing $t$ as in (b) in the statement of the claim finishes the proof.
\end{proof}

Having chosen this matrix $\bR$, we can now analyze the expression \eqref{eq:matrixform}.  
\begin{claim}
There are constants $c,c'$ so that 
with probability at least
\[ 1 -  2\exp \inparen{ - c \cdot \min \inset{ \frac{ \onenorm{\bv}^2\onenorm{\bz}^2 }{ \twonorm{\bv}^2 \twonorm{\bz}^2 } , \frac{ \onenorm{ \bv } \onenorm{ \bz } }{ \infnorm{ \bv}  \twonorm{ \bz} \sqrt{ r \log(r) }} , \frac{ \onenorm{ \bv } \onenorm{ \bz }}{  \infnorm{ \bv} \infnorm{ \bz } r \log(r) } } }
- e \cdot \exp \inparen{ \frac{ - c' r \onenorm{ \bv}^2 }{  \twonorm{ \bv}^2 } },
\]
we have
\[ \fronorm{ \bH \had ( \bX - \bX' )}^2 \geq r \onenorm{ \bv } \onenorm{ \bz }. \]
\end{claim}
\begin{proof}
We break the left hand side up into two terms:
\begin{align*}
\fronorm{ \bH \had ( \bX - \bX' ) }^2 &= 
4 \sum_i v_i \bxi_i^T \bR^T \bD_{\bz} \bR \bxi_i \notag \\
&= 4 \sum_i v_i \bxi_i^T ( \bR^T \bD_{\bz} \bR  - \onenorm{\bz} \bI )\bxi_i + 4 \onenorm{ \bz } \sum_i v_i \twonorm{ \bxi_i }^2 \notag \\
&:= (I) + (II)
\label{eq:twoterms}
\end{align*}

For the first term (I), we will use the Hanson-Wright Inequality (Theorem~\ref{thm:hansonwright}).
In our case, the matrix $\bF$ is a block-diagonal matrix consisting of $d_1$ blocks which are $r \times r$, where the $i$'th block is equal to $4 v_i (\bR^T \bD_{\bz} \bD - \onenorm{\bz} \bI )$.  The Frobenius norm of this matrix is bounded by
\[ \fronorm{\bF}^2 = 16 \sum_i v_i^2 \fronorm{ \bR^T \bD_{\bz} \bD - \onenorm{\bz} \bI }^2 \leq 32 r^2 \twonorm{\bv}^2 \twonorm{ \bz}^2. \]
The operator norm of $\bF$ is bounded by
\[ \opnorm{\bF} = \infnorm{\bv} \infnorm{ \bR^T \bD_{\bz} \bD - \onenorm{\bz} \bI } \leq 2 \infnorm{ \bv} \inparen{ \twonorm{ \bz } \sqrt{ r\log(r)} + \infnorm{ \bz } r \log(r) }. \]
Thus, the Hanson-Wright inequality implies that
\[ \PR{ (I) > t}
\leq 2\exp \inparen{ - c \cdot \min \inset{ \frac{ t^2 }{ 32 r^2 \twonorm{\bz}^2 \twonorm{ \bv }^2 } , \frac{ t}{ \infnorm{ \bv } \inparen{ \twonorm{ \bz} \sqrt{ r \log(r) } + \infnorm{ \bz } r \log(r) } } } }. \]
Plugging in $t = \frac{ r \onenorm{\bz} \onenorm{ \bv } }{2}$, and replacing the constant $c$ with a different constant $c'$, we have
\begin{equation}\label{eq:need1}
 \PR{ (I) > \frac{ r \onenorm{ \bz } \onenorm{ \bv }}{2} }
\leq 2\exp \inparen{ - c' \cdot \min \inset{ \frac{ \onenorm{\bv}^2\onenorm{\bz}^2 }{ \twonorm{\bv}^2 \twonorm{\bz}^2 } , \frac{ \onenorm{ \bv } \onenorm{ \bz } }{ \infnorm{ \bv}  \twonorm{ \bz} \sqrt{ r \log(r) }} , \frac{ \onenorm{ \bv } \onenorm{ \bz }}{  \infnorm{ \bv} \infnorm{ \bz } r \log(r) } } }
\end{equation}

Next we turn to the second term (II).  We write
\[ (II) = 4 \onenorm{ \bz } \sum_i v_i \inparen{ \twonorm{ \bxi_i }^2 - \frac{r}{2} } + 2r \onenorm{ \bz } \onenorm{ \bv } \]
and bound the error term $4 \onenorm{ \bz } \sum_i v_i \inparen{ \twonorm{ \bxi_i }^2 - r/2 }$ with high probability.  Observe that for each $i$, $\twonorm{ \bxi_i }^2 - r/2$ is a mean-zero subgaussian random variable, which satisfies for all $t > 0$ that
\[ \PR{ |\twonorm{ \bxi_i }^2 - r/2| > t } \leq \exp\inparen{ \frac{ -c'' \cdot t^2  }{ r } } \]
for some constant $c''$.  Thus by a version of Hoeffding's inequality (e.g., Proposition 5.10 in \cite{vershynin2010}), for any $t> 0$ we have
\[ \PR{ \inabs{ \sum_i v_i  \twonorm{ \bxi_i }^2 - \frac{ \onenorm{\bv} r }{2} } > t }  \leq e \cdot \exp \inparen{ \frac{ - c''' \cdot t^2 }{ r \twonorm{\bv}^2 } } \]
for some other constant $c'''$.
Thus, 
\begin{align}
\PR{ \inabs{ (II) - 2r \onenorm{ \bz } \onenorm{ \bv } } > \frac{ r \onenorm{ \bv } \onenorm{ \bz} }{2} }
&= \PR{ 4\onenorm{ \bz } \inabs{ \sum_i v_i \inparen{ \twonorm{ \bxi_i }^2 - \frac{r}{2} }} > \frac{ r \onenorm{ \bv } \onenorm{ \bz }} { 2 } } \\
&= \PR{ \inabs{ \sum_i v_i \inparen{ \twonorm{ \bxi_i }^2 - \frac{r}{2} } } > \frac{ r \onenorm{ \bv } }{8 } } \\
&\leq e \cdot \exp \inparen{ \frac{ - c''' r^2 \onenorm{ \bv}^2 }{ 8r \twonorm{ \bv}^2 } }.
\label{eq:need2}
\end{align}

In the favorable case of both \eqref{eq:need1} and \eqref{eq:need2}, we conclude that
\begin{align*}
\fronorm{ \bH \had ( \bX - \bX' ) }^2  &= (I) + (II)\\
&\geq 2r \onenorm{ \bz } \onenorm{ \bv} - \inabs{ (II) - 2r \onenorm{ \bz } \onenorm{ \bv } } - \inabs{ (I) } \\
&\geq r \onenorm{ \bz } \onenorm{ \bv},
\end{align*}
as desired.
\end{proof}
Now a union bound over all of the points in $\cX$ establishes items 1 and 3 of the lemma, along with the observation that $\onenorm{\bz} \onenorm{\bv} = \fronorm{ \bH \had \bA }^2$.
\end{proof}

\section{Case study: when $\lambda$ is small}\label{sec:nogap}
The point of this section is to examine our general bounds from Sections~\ref{sec:upper} and \ref{sec:lower} in the case when the parameter $\lambda = \opnorm{ \bW^{(1/2)} - \bW^{(-1/2)} \had \ind{\Omega} }$ is small.  One case where this happens is the following.

Let $\bW \in \R^{d_1 \times d_2}$ be a rank-1 matrix so so that every entry of $\bW$ satisfies 
\[ W_{ij} \in \left[ \frac{1}{\sqrt{d_1d_2}}, 1 \right]. \]
Now we'd like to consider some $\Omega$ that is ``close'' to $\bW$ in the sense that $\lambda$ is small.  One way to get such an $\Omega$ is to draw it randomly so that $(i,j) \in \Omega$ with probability $W_{ij}$.
As we will show below, in this case $\bW \approx \ind{\Omega}$, and in particular $\lambda$ will be small.\footnote{The reason that we make the assumption that the entries of $\bW$ are not too small is so that $\lambda$ will be small with high probability.  Otherwise, this distribution is not a good case study for the ``small $\lambda$'' case.}  

We emphasize that even though $\Omega$ is drawn at random in this thought experiment, the goal is to understand our bounds for \em deterministic \em sampling matrices $\Omega$.  That is, the upper bounds are still uniform (they hold simultaneously for all appropriate matrices $\bM$), and this model is just a way to generate matrices $\Omega$ so that $\lambda$ is small, to test our uniform bounds on.  We will show that for most $\Omega$ that are close to $\bW$ (in the above sense), our upper and lower bounds from Sections~\ref{sec:upper} and \ref{sec:lower} nearly match.

\subsection{Upper bounds}
In this section we specialize Theorem~\ref{lem:upperrankk} to the case where $\Omega$ is drawn randomly proportional to $\bW$, as discussed above.

We begin with some bounds on the parameters $\lambda$ and $\mu$ in this case.
\begin{lemma}\label{lem:lambdamu}
Let $\bW = \bw \bu^T \in \R^{d_1 \times d_2}$ be a rank-1 matrix 
so that for all $i,j \in [d_1] \times [d_2]$, $W_{ij} \in [ 1/ \sqrt{d_1d_2}, 1 ]$.
Suppose that $\Omega \subseteq [d_1] \times [d_2]$ so that for each $i \in [d_1], j \in [d_2]$, $(i,j) \in \Omega$ with probability $W_{ij}$, independently for each $(i,j)$.  
Then with probability at least $1 - 3/(d_1 + d_2)$ over the choice of $\Omega$, we have
\[ \lambda \leq 2 \sqrt{ d_1 + d_2 } \log(d_1 + d_2) \]
and
\[ \mu\leq 2 \sqrt{ (d_1 + d_2) \log(d_1 + d_2) }, \]
where $\lambda$ and $\mu$ are as in \eqref{eq:lambda} and \eqref{eq:mu}.
\end{lemma}
\begin{proof}
Fix $i \in [d_1]$. Bernstein's inequality yields \newcommand{\omij}{\ensuremath{\ind{(i,j) \in \Omega}}}
\[
\PR{\sum_{j=1}^{d_2} \frac{\omij}{w_i u_j} - d_2 > 2\sqrt{2}(d_1 + d_2)\log(d_1 + d_2)}
\leq \frac{1}{(d_1 + d_2)^2}.
\]
Hence, by taking a union bound, 
\[
\PR{\max_{i \in [d_1]} \sum_{j=1}^{d_2} \frac{\omij}{w_i u_j} > 4(d_1+ d_2) \log (d_1 + d_2)}
\leq \frac{d_1}{(d_1 + d_2)^2} \leq \frac{1}{d_1 + d_2}.
\]
A similar argument gives the bound
\[
\PR{ \max_{j \in [d_2]} \sum_{i=1}^{d_1} \frac{\omij}{w_i u_j} > 4(d_1 + d_2) \log (d_1 + d_2)}
\leq \frac{1}{d_1 + d_2}. 
\]
Combining these two inequalities,
\begin{equation}\label{ineq:mubound-sampleW}
\mu \leq 2\sqrt{(d_1 + d_2) \log (d_1 + d_2)}
\quad
\text{with probability at least $1 - 2/(d_1 + d_2)$.}
\end{equation}
To bound $\lambda = \|\bW^{(-1/2)} \had ( \bW - \ind{\Omega} ) \|$, put
$\gamma_{ij} = (1/\sqrt{w_i u_j})(w_iu_j - \omij)$, $\bX_{ij} = \gamma_{ij} \be_i \be_j^T$, and
write
\[
\bS := \bW^{(-1/2)} \had (\bW - \ind{\Omega}) = \sum_{i=1}^{d_1} \sum_{j=1}^{d_2} \bX_{ij}.
\]
Set $\nu := \max(\|\EE \bS \bS^T \| , \|\EE \bS^T \bS\|)$.
Note that 
\[
\EE \bS\bS^T = \sum_{i=1}^{d_1} \left(\sum_{j=1}^{d_2} \EE \gamma_{ij}^2 \right) \be_i \be_i^T.
\]
Since $\EE \gamma_{ij}^2 = 1 - w_i u_j \leq 1$, the display above gives $\|\EE \bS\bS^T\| \leq d_2$.
Similarly, $\|\EE \bS^T \bS\| \leq d_1$, and so $\nu \leq d_1 + d_2$. 
Furthermore, with probability 1, $|\gamma_{ij}| \leq 2(d_1 d_2)^{1/4} \leq \sqrt{d_1 + d_2}$ so
$\|\bX_{ij}\| \leq \sqrt{d_1 + d_2}$ almost surely. Then, the matrix Bernstein Inequality (Theorem \ref{thm:matrixbernstein}) gives
\[
\PR{\lambda \geq 2\sqrt{d_1 + d_2}\log(d_1 + d_2)} \leq (d_1 + d_2) \exp\left(-\frac{2(d_1 + d_2) \log^2(d_1 + d_2)}{\nu + 2(d_1 + d_2)\log (d_1+ d_2)/3}\right)
\leq \frac{1}{d_1 + d_2}.
\]
Thus,
\begin{equation}\label{ineq:lambdabound-sampleW}
\lambda \leq 2\sqrt{d_1 + d_2} \log (d_1 + d_2),
\quad
\text{with probability at least $1 - 1/(d_1 + d_2)$.}
\end{equation}
\end{proof}

Let $m = \fronorm{ \bW^{(1/2)} }$. It is easy to see that $m = \EE|\Omega|$ as well, and below we show that $|\Omega|$ is very close to $m$ with high probability.
\begin{lemma}\label{claim:Omegam}
Let $m = \fronorm{ \bW^{(1/2)} }$.
There is some constant $C$ so that 
with probability at least $1 - 2\exp(-C\cdot m)$, 
\[ \inabs{ |\Omega| - m } \leq m/4. \]
\end{lemma}
\begin{proof}
We have 
\[ \inabs{ |\Omega| - m } = \inabs{ \sum_{i,j} \left( \ind{ (i,j) \in \Omega} - W_{ij} \right) } = \inabs{ \sum_{i,j} \inparen{ \ind{(i,j) \in \Omega} - \EE \ind{ (i,j) \in \Omega } } }, \]
which is the sum of mean-zero independent random variables.  By Bernstein's inequality, 
\[ \PR{ \inabs{|\Omega| - m } \geq m/4 } \leq 2 \exp \inparen{ \frac{ -m^2/32 }{ \twonorm{ \bw }^2 \twonorm{ \bu }^2 + m/12 } }. \]
Now using the assumption that $\bw$ and $\bu$ are flat, we have
\[ \twonorm{ \bw }^2 \twonorm{ \bu }^2 \leq \frac{(C')^4}{d_1d_2} \onenorm{ \bw }^2 \onenorm{ \bu }^2 = \frac{ (C')^4 m^2 }{d_1d_2 } \leq (C')^4 m, \]
which proves the claim after choosing
$ C = \frac{ 1 } { 32 ( (C')^4 + 12 }.$
\end{proof}

With these computations out of the way, we may apply Theorems~\ref{lem:upperrankk} and \ref{lem:uppermaxball} in this setting.  
Theorem~\ref{thm:upperOmega} follows immediately from Theorem~\ref{lem:upperrankk} and Lemmas~\ref{lem:lambdamu} and \ref{claim:Omegam}.
\begin{theorem}\label{thm:upperOmega}
Let $\bW = \bw \bu^T \in \R^{d_1 \times d_2}$ be a rank-1 matrix 
so that for all $i,j \in [d_1] \times [d_2]$, $W_{ij} \in [1/\sqrt{d_1d_2}, 1]$.
Suppose that $\Omega \subseteq [d_1] \times [d_2]$ so that for each $i \in [d_1], j \in [d_2]$, $(i,j) \in \Omega$ with probability $W_{ij}$, independently for each $(i,j)$. 
Then with probability at least $1 - 4/(d_1 + d_2)$ over the choice of $\Omega$, the following holds.

There is an algorithm $\cA$ so that for any rank-r matrix $\bM$ with $\infnorm{\bM} \leq \beta$,
$\mathcal{A}$ returns $\hat{M} = \cA( \bM_\Omega + \bZ_\Omega)$
so that with probability at least $1 - 1/d$ over the choice of $\bZ$,
\[
\frac{\fronorm{ \bW^{(1/2)} \had ( \bM - \hat{ \bM } ) }}{ \fronorm{ \bW^{(1/2)} } }  \leq 8 \beta \sqrt{ \frac{ r^2 (d_1 + d_2)  } {|\Omega|}} \log(d_1 + d_2) +
16\sigma  \sqrt{\frac{ r(d_1 + d_2)}{|\Omega|} }\log(d_1 + d_2 ).
\]
\end{theorem}
\begin{proof}
Plugging in Lemma~\ref{lem:lambdamu} to Theorem~\ref{lem:upperrankk} shows that
\[
\fronorm{ \bW^{(1/2)} \had ( \bM - \hat{ \bM } ) }  \leq  4r\sqrt{d_1 + d_2} \log(d_1 + d_2) \infnorm{ \bM } +
8\sigma  \sqrt{r(d_1 + d_2)}\log(d_1 + d_2 ).
\]
The result follows by using $\infnorm{\bM} \leq \beta$, and by normalizing and using Lemma~\ref{claim:Omegam} to replace $\fronorm{ \bW^{(1/2)} }$ with $|\Omega|$.
\end{proof}

Similarly,
Theorem~\ref{thm:upperOmegaApprox} below follows immediately from Theorem~\ref{lem:uppermaxball} and Lemmas~\ref{lem:lambdamu} and Lemma~\ref{claim:Omegam}.
\begin{theorem}\label{thm:upperOmegaApprox}
There is some constant $C$ so that the following holds.
Let $\bW = \bw \bu^T \in \R^{d_1 \times d_2}$ be a rank-1 matrix 
so that for all $i,j \in [d_1] \times [d_2]$, $W_{ij} \in [1/\sqrt{d_1d_2}, 1]$.
Suppose that $\Omega \subseteq [d_1] \times [d_2]$ so that for each $i \in [d_1], j \in [d_2]$, $(i,j) \in \Omega$ with probability $W_{ij}$, independently for each $(i,j)$.  
Then with probability at least $1 - 4/(d_1 + d_2)$ over the choice of $\Omega$, the following holds.

There is an algorithm $\cA$ so that for $d_1 \times d_2$ matrix $\bM \in \beta \sqrt{r} B_{\max}$,
$\mathcal{A}$ returns $\hat{M} = \cA( \bM_\Omega + \bZ_\Omega)$
so that with probability at least $1 - 1/(d_1 + d_2)$ over the choice of $\bZ$,
\[
\frac{ \fronorm{ \bW^{(1/2)} \had ( \bM - \hat{ \bM } ) } }{ \fronorm{ \bW^{(1/2) } }} \leq 
C\beta \inparen{ \frac{ r^2 (d_1 + d_2) }{m} }^{1/4} \log^{1/2}(d_1 + d_2)  + C \sqrt{\beta \sigma} \inparen{ \frac{ r(d_1 + d_2) }{m } }^{1/4} \log^{1/4}(d_1 + d_2). \]
\end{theorem}

\subsection{Lower bound for exactly rank $r$ matrices}
In this section, we will prove a lower bound in the case where $\Omega$ is drawn proportionally to $\bW$.
We begin with a warm-up result for ``flat'' weight matrices $\bW$.

\begin{lemma}[Lower bound for low-rank matrices when $\bW$ is flat and $\Omega \sim \bW$]\label{lem:flatLB}
Let $\bW = \bw \bu^T \in \R^{d_1 \times d_2}$ be a rank-1 matrix with strictly positive entries and with $\infnorm{ \bW } \leq 1$.  
Suppose that there is some constant $C'$ so that 
\[\max_i |w_i| \leq C' \min_i |w_i| \qquad \text{and} \qquad \max_i |u_i| \leq C' \min_i |u_i|. \]
Suppose that $\Omega \subseteq [d_1] \times [d_2]$ so that for each $i \in [d_1], j \in [d_2]$, $(i,j) \in \Omega$ with probability $W_{ij}$, independently for each $(i,j)$.  
Then with probability at least $1 - \exp(-C\cdot m)$ over the choice of $\Omega$, the following holds:

Let $\sigma > 0$, let $0 < r < \inparen{ \frac{ \min\{d_1,d_2\} }{ \log( d_1d_2 ) } }^{1/3}$, and let $K \subset \R^{d_1 \times d_2}$ be the cone of rank $r$ matrices. 
For any algorithm $\mathcal{A}: \R^\Omega \to \R^{d_1 \times d_2}$ that takes as input $\bX_\Omega + \bZ_\Omega$ and outputs a guess $\hat{\bX}$ for $\bX$, for $\bX \in K \cap  \beta B_\infty$ and $Z_{ij} \sim \gN(0, \sigma^2)$ there is some $\bM \in K \cap  \beta B_\infty$ so that
\[ \frac{1}{ \fronorm{ \bW^{(1/2)} }} \cdot \fronorm{ \bW^{(1/2)} \had \inparen{ \cA ( \bM_\Omega + \bZ_\Omega ) - \bM } } \geq c \min \inset{ \sigma \sqrt{ \frac{ r\max\{d_1, d_2 \} }{|\Omega| }}, \frac{\beta}{ \sqrt{ \log(d_1d_2) } } } \]
with probability at least $1/2$ over the randomness of $\mathcal{A}$ and the choice of $\bZ$.  Above, $c,C$ are constants which depend only on $C'$.
\end{lemma}

\begin{proof}
Suppose without loss of generality that $d_1 \geq d_2$,
and that $\log d_1$ and $\log d_2$ are integers (if not, replace them
by their floors.)
Let $m = \fronorm{ \bW^{(1/2)} } = \onenorm{ \bw } \onenorm{ \bu }$, so that $\EE |\Omega| = m$.

Next, we instantiate Lemma~\ref{lem:net} with $\bH = \bW^{(1/2)}$ and $\bA = \mathbf{1} \mathbf{1}^T$.  In the language of that lemma, we have
\[ \bv = (\bh \had \ba)^{(2)} = \bh^{(2)} = \bw \]
\[ \bz = (\bg \had \bb)^{(2)} = \bg^{(2)} = \bu. \]
Let $\cX$ be the net guaranteed by Lemma~\ref{lem:net}.  We have
\begin{equation} \label{eq:maxXOmega}
 \max_{ \bX \in \cX } \fronorm{ \bX_\Omega } \leq \sqrt{ cr } \fronorm{ \bA_\Omega } = \sqrt{ cr|\Omega| } = \sqrt{ c' r m } 
\end{equation}
for some constant $c'$, using Lemma \ref{claim:Omegam}.
We also have
\begin{equation}\label{eq:dist}
\fronorm{  \bW^{(1/2)}  \had ( \bX - \bX' ) } \geq \sqrt{ r } \fronorm{ \bW^{(2)} } = \sqrt{ rm }
\end{equation}
for all $\bX \neq \bX' \in \cX$,
using the definition of $m$.  We have
\begin{equation}\label{eq:infnorm}
\max_{ \bX \in \cX } \infnorm{ \bX } \leq c \sqrt{ r \log(d_1d_2) }.
\end{equation}
And finally, again using the assumption that $\bw$ and $\bu$ are flat, the size of the net is
\begin{align*} 
N &= 2e \exp \inparen{  c \cdot \min \inset{ \frac{ \onenorm{\bv}^2\onenorm{\bz}^2 }{ \twonorm{\bv}^2 \twonorm{\bz}^2 } , \frac{ \onenorm{ \bv } \onenorm{ \bz } }{ \infnorm{ \bv}  \twonorm{ \bz} \sqrt{ r \log(r) }} , \frac{ \onenorm{ \bv } \onenorm{ \bz }}{  \infnorm{ \bv} \infnorm{ \bz } r \log(r) }  , \frac{ r^2 \onenorm{ \bv}^2 }{ r \twonorm{ \bv}^2 } }} \\
&\geq 2e \exp \inparen{ \frac{c}{(C')^4} \min \inset{ d_1d_2, \frac{ d_1 \sqrt{d_2} }{ \sqrt{ r\log(r) } }, \frac{ d_1d_2 }{ r \log(r) } , r d_1 }} \\
&\geq \exp(C'' rd_1),
\end{align*}
where in the last line we have used the assumption that $r$ is not too large compared to $d_2$, and a suitable choice of $C''$ which depends on $c$ and $C'$.

Now we can use this net in Lemma~\ref{lem:fano}.  
We choose
\[ \kappa = \min \inset{ c'' \sigma \sqrt{ \frac{ d_1 }{m}}  , \frac{ \beta }{ c \sqrt{ r \log(d_1d_2 )} }  }, \]
where $c'' = \frac{1}{4} \sqrt{ C''/c'} $ depends on previous constants.

Observe that by \eqref{eq:infnorm} we have 
\[ \frac{ \sigma \sqrt{ \log|\cX| } }{ 4 \max_{ \bX \in \cX } \fronorm{ \bX_\Omega } } 
\geq \frac{\sigma \sqrt{ C'' rd_1} }{ 4 \sqrt{ c' r m  } } \geq \kappa, \]
so this is a legitimate choice of $\kappa$ in Lemma~\ref{lem:fano}.

Next, we verify that $\kappa \cX \subseteq K_r \cap \beta B_\infty$.  Indeed, we have
\[ \kappa \max_{\bX \in \cX} \infnorm{ \bX } \leq \kappa c\sqrt{ r \log(d_1d_2) } \leq \beta, \]
so $\kappa \cX \subseteq \beta B_\infty$, and every element of $\cX$ has rank $r$ by construction.

Then Lemma~\ref{lem:fano} concludes that if $\mathcal{A}$ must work on $K_r \cap \beta B_\infty$, then there is a matrix $\bM \in K_r \cap  \beta  B_\infty$ so that
\begin{align*}
 \fronorm{ \bW^{(1/2)} \had ( \cA( \bM_\Omega + \bZ_\Omega ) - \bM ) } &\geq \frac{ \kappa }{2 } \min_{ \bX \neq \bX' \in \cX } \fronorm{ \bW^{(1/2)} \had ( \bX - \bX' )  } \\
& \geq \frac{1}{2} \min \inset{ c'' \sigma \sqrt{ \frac{d_1}{m} } , \frac{ \beta }{ c \sqrt{ r \log(d_1d_2) } } } \sqrt{rm} \\
&= \frac{1}{2} \min \inset{ c'' \sigma \sqrt{ rd_1 } , \frac{ \beta \sqrt{m} }{ c \sqrt{ \log(d_1d_2) } } },
\end{align*}
using \eqref{eq:dist}.
Normalizing by $\fronorm{ \bW^{(1/2) } }$ and applying Lemma~\ref{claim:Omegam} completes the proof.
\end{proof}

Next, we use Lemma \ref{lem:flatLB} to prove a bound that does not require that the $\bW$ be flat; that is, Theorem~\ref{thm:simOmegaLB} below is similar to Lemma~\ref{lem:flatLB}, but we will not require the ``flatness'' assumption that $\max_i |w_i| \leq C' \min_i |w_i|$ or $\max_i|u_i| \leq C' \min_i |u_i|$.  However, we do have to impose an additional restriction that the entries of $\bu$ and $\bw$ are not smaller than $1 / \sqrt{d_1}$ and $1/\sqrt{d_2}$ respectively; this is the same restriction we have for the upper bounds.

Our final lower bound for the case when $\lambda$ is small is the following.
\begin{theorem}[Lower bound for rank-$k$ matrices when $\Omega \sim \bW$]\label{thm:simOmegaLB}
Let $\bW = \bw \bu^T \in \R^{d_1 \times d_2}$ be a rank-1 matrix with $\infnorm{ \bW } \leq 1$, so that 
\[ \infnorm{ \bw^{(-1)} } \leq \sqrt{ d_1 } \qquad \text{and}  \qquad \infnorm{ \bu^{(-1)} } \leq \sqrt{ d_2 }.\]   Let $d = \sqrt{ d_1d_2 }.$
Suppose that $\Omega \subseteq [d_1] \times [d_2]$ so that for each $i \in [d_1], j \in [d_2]$, $(i,j) \in \Omega$ with probability $W_{ij}$, independently for each $(i,j)$.  
Let $m = \onenorm{ \bw } \onenorm{ \bu }$, so that $\EE |\Omega| = m$.
Then with probability at least $1 - \exp(-C\cdot m)$ over the choice of $\Omega$, the following holds:

Let $\sigma, \beta > 0$, let $0 < r < \inparen{ \frac{ \min\{d_1,d_2\} }{ \log^2( d ) } }^{1/3}$, and let $K_r \subset \R^{d_1 \times d_2}$ be the cone of rank $r$ matrices. 
For any algorithm $\mathcal{A}: \R^\Omega \to \R^{d_1 \times d_2}$ that takes as input $\bX_\Omega + \bZ_\Omega$ and outputs a guess $\hat{\bX}$ for $\bX$, for $\bX \in K \cap \beta B_\infty$ and $Z_{ij} \sim \gN(0, \sigma^2)$ there is some $\bM \in K_r \cap \beta B_\infty$ so that
\begin{align*}
\frac{1}{ \fronorm{ \bW^{(1/2)} } }\fronorm{ \bW \had ( \cA( \bM_\Omega + \bZ_\Omega ) ) } 
&\geq c \min \inset{ \sigma \sqrt{ \frac{ r \max\{d_1,d_2\} } {m \log(d) } }, \beta \sqrt{ \frac{ d }{ m \log^3(d)  } }}. 
\end{align*}
with probability at least $1/2$ over the randomness of $\mathcal{A}$ and the choice of $\bZ$.  Above, $c,C$ are constants which depend only on $C'$.
\end{theorem}

\begin{proof}
Suppose without loss of generality that $d_1 \geq d_2$.  Write
\[ \bw = (\bw_1, \bw_2, \ldots, \bw_{\log(d_1)/2}) \qquad \text{where} \qquad \bw_i \in \RR^{s_i}, \]
so that all entries of $\bw_i$ are in $[2^{-i}, 2^{1-i}]$ for all $i$.  (Here, we assume without loss of generality that the coordinates of $\bw$ are arranged in decreasing order).
Notice that this is possible because the $\max_i |w_i| \leq 1$ and $\min_i |w_i| \geq \frac{1}{\sqrt{d_1}}$ by assumption.
Similarly write
\[ \bu = (\bu_1, \bu_2, \ldots, \bu_{\log(d_2)/2}) \qquad \text{where} \qquad \bu_i \in \RR^{t_i}, \]
so that all entries of $\bu_i$ are in $[2^{-i}, 2^{1-i}]$.

Now there is some $i$ and $j$ so that
\[ s_i \geq \frac{ 2d_1 }{ \log(d_1) } \qquad t_j \geq \frac{ 2d_2 }{ \log(d_2 ) }. \]
Now consider $\tilde{\bW} = \bw_i \bu_j^T \in \RR^{s_i \times t_j}$, and let $\tilde{\Omega}$ be the restriction of $\Omega$ the the $s_i$ rows and $t_j$ columns corresponding to $\tilde{W}$.
Notice that $\infnorm{ \bw_i } \leq 2 \infnorm{ \bw_i^{(-1)} }^{-1}$ and $\infnorm{ \bu_j } \leq 2 \infnorm{ \bu_j^{(-1)} }^{-1}$ by definition.
Now we may apply Lemma~\ref{lem:flatLB} to the problem of recovering an $s_i \times t_j$ matrix, and conclude that there is some matrix $\bM \in \RR^{s_i \times t_j}$ so that 
\[
\fronorm{ \tilde{\bW} \had ( \cA( \bM_{\tilde{\Omega}} + \bZ_{\tilde{\Omega}} ) ) } \geq c \min \inset{ \sigma \sqrt{ \frac{ r \max\{d_1,d_2\} } {\log(d_1d_2) } }, \frac{ \fronorm{ \tilde{\bW}^{(1/2)} } \beta }{ \sqrt{ \log(d_1d_2) } } }. \]
Now, we observe that 
\[ \fronorm{ \tilde{\bW}^{(1/2)} } = \sqrt{ \onenorm{  \bw_i } \onenorm{ \bu_j } } \geq \sqrt{ \frac{ 4\sqrt{d_1d_2} }{ \log(d_1) \log(d_2) }} , \]
using the fact that both $\bw_i$ and $\bu_j$ have entries no smaller than $1/\sqrt{d_1}$ and $1/\sqrt{d_2}$ respectively.  
Thus, normalizing appropriately, we conclude
\begin{align*}
\frac{1}{ \fronorm{ \bW^{(1/2)} } }\fronorm{ \bW \had ( \cA( \bM_\Omega + \bZ_\Omega ) ) } 
&\geq
\frac{1}{ \fronorm{ \bW^{(1/2)} } }\fronorm{ \tilde{\bW} \had ( \cA( \bM_{\tilde{\Omega}} + \bZ_{\tilde{\Omega}} ) ) } \\
&\geq c \min \inset{ \sigma \sqrt{ \frac{ r \max\{d_1,d_2\} } {m \log(d_1d_2) } }, \beta \sqrt{ \frac{ \sqrt{d_1d_2} }{ m \log(d_1) \log(d_2) \log(d_1d_2) } } }. 
\end{align*}
The theorem follows after simplifying with the definition $d = \sqrt{ d_1d_2 }$.
\end{proof}

\subsection{Lower bound for approximately rank $r$ matrices}
In Theorem~\ref{thm:upperOmega} that for rank $r$ matrices it was sufficient for $m = \fronorm{ \bW^{(1/2)}} = \EE |\Omega|$ to grow like $1/\delta^2$ in order to guarantee that the per-entry normalized error is at most $\delta$.  However, our upper bound for approximately rank-$r$ matrices (Theorem \ref{thm:upperOmegaApprox}) requires $m$ to grow like $1/\delta^4$.  In this section, we show that this dependence is necessary.

As with the previous section, we begin by focusing on flat matrices.  Unfortunately, our lower bounds do not seem to extend in the same way to non-flat matrices without losing the correct dependence on the error.
Thus, we state a result in the approximately low-rank setting only for flat matrices.

\begin{theorem}[Lower bound approximately low-rank matrices when  $\bW$ is flat  and $\Omega \sim W$]\label{lem:flatLBapprox}
Let $\bW = \bw \bu^T \in \R^{d_1 \times d_2}$ be a rank-1 matrix with $\infnorm{ \bW } \leq 1$.  
Suppose that there is some constant $C'$ so that 
\[\max_i |w_i| \leq C' \min_i |w_i| \qquad \text{and} \qquad \max_i |u_i| \leq C' \min_i |u_i|. \]
Suppose that $\Omega \subseteq [d_1] \times [d_2]$ so that for each $i \in [d_1], j \in [d_2]$, $(i,j) \in \Omega$ with probability $W_{ij}$, independently for each $(i,j)$.   Let $m = \EE|\Omega|$.
Then with probability at least $1 - \exp(-C\cdot m)$ over the choice of $\Omega$, the following holds:

Let $\sigma, \beta > 0$, and suppose that
\[ \frac{ \beta }{ \sigma } \leq \frac{ \min\{d_1,d_2\}^{1/3} \cdot \max\{d_1, d_2\}^{1/2} }{ \sqrt{rm} \log^{2/3}(d_1d_2) }. \]
Let $K = \beta \sqrt{r} B_{\max}$ be the max-norm ball of radius $\beta \sqrt{r}$.

For any algorithm $\mathcal{A}: \R^\Omega \to \R^{d_1 \times d_2}$ that takes as input $\bX_\Omega + \bZ_\Omega$ and outputs a guess $\hat{\bX}$ for $\bX$, for $\bX \in K$ and $Z_{ij} \sim \gN(0, \sigma^2)$ there is some $\bM \in K$ so that
\[ \frac{1}{ \fronorm{ \bW^{(1/2)} }} \cdot \fronorm{ \bW^{(1/2)} \had \inparen{ \cA ( \bM_\Omega + \bZ_\Omega ) - \bM } } \geq c \sqrt{ \sigma \beta } \inparen{ \frac{ d }{ m } }^{1/4}, \]
with probability at least $1/2$ over the randomness of $\mathcal{A}$ and the choice of $\bZ$.  Above, $c,C$ are constants which depend only on $C'$.
\end{theorem}

\begin{proof}
The proof proceeds similarly to that of Lemma~\ref{lem:flatLB}, except we will use Lemma~\ref{lem:net} with a different choice of $r$: intuitively, we will choose a net consisting of higher-rank matrices that have smaller $\ell_\infty$ norm, but still have bounded max-norm.

Choose a parameter 
\[ s = \left \lfloor \frac{ \beta }{ \sigma } \sqrt{ \frac{ mr }{ d_1 } } \right \rfloor. \] 

Without loss of generality, suppose that $d_1 \geq d_2$.  Let $m = \fronorm{ \bW^{(1/2)} }^2 = \onenorm{ \bw } \onenorm{ \bu } = \EE|\Omega|$.  We will instantiate Lemma~\ref{lem:net} with $\bH = \bW^{(1/2)}$, $\bA = \mathbf{1} \mathbf{1}^T$, and as in Lemma~\ref{lem:flatLB}, this choice implies that $\bv = \bw$, $\bz = \bu$ in the language of Lemma~\ref{lem:net}.  However, unlike in the Lemma~\ref{lem:flatLB}, we will instantiate Lemma~\ref{lem:net} with $s$ in the place of the ``$r$'' parameter.  Following the same analysis as before, this yields a net of size at least $N \geq \exp( C'' sd_1 )$, with
\[ \max_{ \bX \in \cX } \fronorm{ \bX_\Omega } \leq \sqrt{ c' sm } \qquad \fronorm{ \bW^{(1/2)} \had (\bX - \bX') } \geq \sqrt{ sm } \qquad \max_{ \bX \in \cX } \maxnorm{ \bX } \leq s. \]
Above, we have used our condition on $\beta/\sigma$ to ensure that $s \leq \inparen{ \frac{ \min\{d_1,d_2\} }{\log(d_1d_2) } }^{1/3}$, the analog of our condition on $r$ in Lemma~\ref{lem:flatLB}.
Now choose
\[ \kappa = \min\inset{ c'' \sigma \sqrt{ \frac{ d_1 }{m} }, \frac{ \beta \sqrt{r} }{s} }. \]
As before, we have
\[ \frac{ \sigma \sqrt{ \log|\cX| } }{ 4 \max_{ \bX \in \cX} \fronorm{ \bX_\Omega } } \geq \frac{ \sigma \sqrt{ C'' s d_1 } }{ 4 \sqrt{ c' s m } } \geq \kappa, \]
so this is a legitimate choice for $\kappa$.  We also have 
\[ \kappa \cX \subseteq \kappa s B_{\max} \subseteq \beta \sqrt{r} B_{\max}, \]
and so this net does indeed live in $K$.

Thus, Lemma~\ref{lem:fano} concludes that if $\cA$ must work on $K$, then there is a matrix $\bM \in K$ so that 
\begin{align*}
\fronorm{ \bW^{(1/2)} \had( \cA( \bM_\Omega + \bZ_\Omega ) - \bM ) } &\geq \frac{ \kappa }{2} \min_{\bX \neq \bX' \in \cX} \fronorm{\bW^{(1/2)} \had (\bX - \bX') } \\
&\geq \frac{1}{2} \min\inset{ c'' \sigma \sqrt{ \frac{ d_1 }{m} } , \frac{ \beta \sqrt{r} }{s} } \sqrt{ sm }\\
&= \frac{1}{2} \min \inset{ c'' \sigma \sqrt{ d_1 s } , \beta \sqrt{ \frac{ rm }{s} } }.
\end{align*}
Now normalizing by $\fronorm{ \bW^{(1/2)} } = \sqrt{m}$ and plugging in our choice of $s$, we have
\begin{align*}
\frac{\fronorm{ \bW^{(1/2)} \had( \cA( \bM_\Omega + \bZ_\Omega ) - \bM ) }}{ \fronorm{ \bW^{(1/2)}} }
&\geq c''' \cdot \min\inset{ \sigma \sqrt{ \frac{ d_1 s }{ m } }, \beta \sqrt{ \frac{ r }{s } } } \\
&= c''' \sqrt{ \sigma \beta } \inparen{ \frac{ d_1 r }{ m } }^{1/4},
\end{align*}
as desired.

\end{proof}

\subsection{Application: proportional sampling}\label{sec:nogapapp}
In this section, we show how to apply Theorem~\ref{thm:upperOmegaApprox} to recover results similar to some of those in \cite{CBSW15}.  In that work, the authors show how to do matrix completion on \em coherent \em low-rank matrices, assuming that the observations are drawn from an appropriately biased distribution: more precisely, they show that if entries are sampled with a probability that is based on the \em leverage scores \em of the matrix,\footnote{The leverage scores of a matrix $\bM \in \mathbb{R}^{d_1 \times d_2}$ with SVD $\bU \mathbf{\Sigma} \bV^T$ are the values $\|\be_i^T \bU\|_2$ and $\|\be_j^T \bV\|_2$ for $i \in [d_1], j \in [d_2]$} then matrix completion is possible.
That work proposes a two-stage matrix completion scheme, which first samples from uniformly random entries to estimate the SVD $\bM \approx \tilde{\bU} \tilde{ \mathbf{\Sigma} } \tilde{\bV}^T$, and then uses this to approximate the leverage scores and sample according to this approximation.

Even though the focus of this work is \em deterministic \em matrix completion, Theorem~\ref{thm:upperOmegaApprox} does apply to the randomized setting as well, and this allows us to recover results similar to those of \cite{CBSW15}.  More precisely, we show below that given \em any \em matrix $\bX$, if we define $\bW$ appropriately, then we can ensure that $\bM = \bW^{(-1/2)} \had \bX$ has small max norm, and so by sampling proportional to the entries of $\bW$, we can use Theorem~\ref{thm:upperOmegaApprox} to obtain error bounds on  $$\| \bW^{(1/2)} \had ( \bM - \hat{\bM} ) \|_F = \| \bX - \bW^{(1/2)} \had \hat{\bM} \|_F.$$
Thus, if we have a good enough estimate of $\bX$ to do the sampling, we may use our algorithm from Theorem~\ref{thm:upperOmegaApprox} to obtain $\hat{\bM}$ and then set $\hat{\bX} = \bW^{(1/2)} \had \hat{\bM}$.

The sampling procedure that results ends up being slightly different than the leverage scores, but it can still be approximated in a two-stage algorithm using an approximate SVD $\bM \approx \tilde{\bU} \tilde{ \mathbf{\Sigma} } \tilde{\bV}^T$.  The main difference between this and the corresponding theorem in \cite{CBSW15} is that their work applies to exactly rank-$r$ matrices, while the result below applies to any matrix, and is stated in terms of  $( \nucnorm{\bX} \cdot \fronorm{ \bX }^{-1} )^2$, a proxy for the rank.

\begin{corollary}\label{cor:leveragescores}
For a $d_1 \times d_2$ matrix $\bX$ with SVD $\bX = \bU \mathbf{\Sigma} \bV^T$ and an integer $m$, define
\[ W_{i,j} = \frac{ \| \be_i^T \bU \mathbf{\Sigma}^{(1/2)}\|^2_2 \| \be_j^T \bV \mathbf{\Sigma}^{(1/2)} \|^2_2 }{ \nucnorm{ \bX }^2 } m. \]
There is a randomized algorithm $\mathcal{A}$ with query access to a $d_1 \times d_2$ matrix $\bX$ so that
the following holds.

Suppose that $\bX$ is any matrix so that $W_{i,j} \in [1/\sqrt{d_1d_2}, 1]$.  Then with probability at least $1 - \frac{4}{d_1 + d_2}$, 
 $\mathcal{A}$ queries at most 
\[ m \geq \frac{ C }{\eps^4} \inparen{ \frac{ \nucnorm{\bX} }{ \fronorm{\bX} } }^4 (d_1 + d_2 ) \log^2( d_1 + d_2 ) , \]
entries of $\bX$, and returns $\hat{\bX}$ so that
\[ \fronorm{ \bX - \hat{\bX} } \leq \eps \fronorm{ \bX }, \]
where $C$ is some absolute constant.
Moreover, $\mathcal{A}$ queries entries of $\bX$ independently, querying $X_{i,j}$ with probability proportional to $W_{i,j}$.
\end{corollary}

\begin{remark}
We remark that Corollary~\ref{cor:leveragescores} suggests one potential application of our work that may be an interesting future direction to analyze and develop rigorously. Indeed, a natural approach to take advantage of this result would be to truncate the $W_{i,j}$ entries if they happen to be larger than $1$, leading to a possible ``two-stage'' scheme within this framework.
Also, we note that our results are incomparable to those in \cite{CBSW15},
 as discussed in Section~\ref{sec:related}, the focus of \cite{CBSW15} is instead to  recover coherent matrices using adapted sampling patterns.
\end{remark}

Before we prove the corollary, we interpret it.  We observe that the ratio $(\nucnorm{ \bX }\cdot\fronorm{ \bX }^{-1})^2$ is a proxy for the rank of $\bX$, in the sense that if $\bX$ is actually rank $r$, then this quantity is at most $r$
by the Cauchy-Schwarz inequality. 
Thus, Corollary~\ref{cor:leveragescores} says that if $\bX$ is not \em too \em coherent (in the sense that the $W_{i,j}$ are between $1/\sqrt{d_1d_2}$ and $1$), there is some way to sample $m \approx r^2\cdot (d_1 + d_2)$ entries of matrix $\bX$ with $(\nucnorm{\bX} \fronorm{\bX}^{-1})^2 \leq r$ and use these entries to accurately reconstruct $\bX$. 

\begin{proof}
Let
$ \bX = \bU \mathbf{\Sigma} \bV^T $
be the SVD of $\bX$.
Then let $m$ be as in the statement of the corollary, and let $\bW$ be the matrix so that 
\[ W_{i,j} = \frac{ \| \be_i^T \bU \mathbf{\Sigma}^{(1/2)}\|^2_2 \| \be_j^T \bV \mathbf{\Sigma}^{(1/2)} \|^2_2 }{ \nucnorm{ \bX }^2 } m. \]
Notice that $\bW$ is rank $1$.  Let
\[ \bM = \bW^{(-1/2)} \had \bX. \]
Then we may write
\[ \bM = \bW^{(-1/2)} \had \bU \mathbf{\Sigma} \bV^T = \frac{ \nucnorm{\bX}} { \sqrt{m} } \bD_1 \bU \mathbf{\Sigma} \bV^T \bD_2, \]
where $\bD_1$ is a $d_1 \times d_1$ diagonal matrix with $(i,i)$-th entry $\frac{1}{\|\be_i^T \bU \mathbf{\Sigma}^{(1/2)}\|_2}$ and $\bD_2$ is a $d_2 \times D_2$ diagonal matrix with $(j,j)$-th entry $\frac{1}{\|\be_j^T \bV \mathbf{\Sigma}^{(1/2)} \|_2}$.  By construction, every row of $\bD_1 \bU \mathbf{\Sigma}^{(1/2)}$ has $\ell_2$-norm at most $1$, and every row of $\bD_2 \bV \mathbf{\Sigma}^{(1/2)}$ has $\ell_2$ norm at most $1$, and so we conclude that
\[ \|\bM\|_{\max} \leq \| \bW^{(-1/2)} \had \bX \|_{\max} = \frac{ \nucnorm{ \bX } }{\sqrt{m}}. \]

Now, we define the algorithm $\mathcal{A}$ as follows: sample $(i,j) \in [d_1] \times [d_2]$ with probability $W_{i,j}$, independently for each $(i,j)$.  
Then estimate $\hat{\bX}$ is given by
\[ \hat{\bX} = \bW^{(1/2)} \had \hat{\bM}, \]
where $\hat{\bM}$ is the estimate guaranteed by Theorem~\ref{thm:upperOmegaApprox}.\footnote{Looking into the proof of that theorem, we see that we should take
\[ \hat{ \bX } = \argmin_{ \|\bZ\|_{\max} \leq \frac{\nucnorm{ \bX }}{\sqrt{m}} } \| \bZ - \bW^{(-1)} \had \bX_\Omega \|. \] 
}

First, we observe that the expected number of samples taken by $\mathcal{A}$ is
\begin{align*} 
\sum_{i,j} W_{i,j} &= \sum_{i,j} \frac{ \| \be_i^T \bU \mathbf{\Sigma}^{(1/2)} \|^2_2 \| \be_j^T \bV\mathbf{\Sigma}^{(1/2)}  \|^2_2 }{ \nucnorm{ \bX }^2 } m \\
&= \frac{m \cdot \|\bU\mathbf{\Sigma}^{(1/2)} \|_F^2 \|\bV\mathbf{\Sigma}^{(1/2)} \|_F^2 }{\nucnorm{ \bX }^2} \\
&\leq m,
\end{align*}
where above we used the fact that 
\[ \|\bU\mathbf{\Sigma}^{(1/2)} \|_F^2 \|\bV\mathbf{\Sigma}^{(1/2)} \|_F^2 = \nucnorm{ \bX }^2.\]  
To see this, notice that
\begin{align*}
 \|\bU\mathbf{\Sigma}^{(1/2)} \|_F^2 \|\bV\mathbf{\Sigma}^{(1/2)} \|_F^2 
&= \inparen{ \sum_{\ell} |\sigma_\ell| \twonorm{ \bU \be_\ell }^2 } \inparen{ \sum_r |\sigma_r|\twonorm{ \bV \be_r }^2 } \\
&= \sum_{\ell, r} |\sigma_\ell| |\sigma_r| \\
&= \inparen{ \sum_\ell |\sigma_\ell| }^2\\
&= \nucnorm{ \bX }^2, 
\end{align*}
using the fact that $\bU$ and $\bV$ are orthogonal matrices and hence have columns of unit norm.

Next, we observe that by Theorem~\ref{thm:upperOmegaApprox} (with $\sigma = 0$), we have
\begin{align*}
 \|\bX - \hat{\bX}\|_F 
&=  \| \bW^{(1/2)} \had (\bM - \hat{\bM} ) \|_F \\
&\leq  C\|\bW^{(1/2)}\|_F  \|\bM\|_{\max} \inparen{ \frac{ d_1 + d_2 }{m} }^{1/4} \sqrt{ \log( d_1 + d_2) } \\
&\leq C \sqrt{m} \cdot \frac{ \nucnorm{\bX}}{ \sqrt{m} } \cdot \inparen{ \frac{ d_1 + d_2 }{m} }^{1/2} \sqrt{\log(d_1 + d_2) } \\
&= C \cdot \fronorm{ \bX } \inparen{ \frac{ \nucnorm{\bX} }{\fronorm{\bX} } } \inparen{ \frac{ d_1 + d_2 }{m} }^{1/4} \sqrt{\log(d_1 + d_2) }.
\end{align*}
In particular, if 
\[ m \geq \frac{ C^4 \inparen{ \frac{ \nucnorm{\bX} }{ \fronorm{\bX} } }^4 (d_1 + d_2 ) \log^2( d_1 + d_2 ) }{\eps^4 }, \]
then
\[ \| \bX - \hat{ \bX } \|_F \leq \eps \cdot \fronorm{ \bX }, \]
as claimed.  After adjusting the constant $C$ this implies the corollary.
\end{proof}

\section{Case study: when $\lambda$ is large}\label{sec:gap}
The point of this section is to examine our general bounds from Sections~\ref{sec:upper} and \ref{sec:lower} in the case when the parameter $\lambda = \opnorm{ \bW^{(1/2)} - \bW^{(-1/2)} \had \ind{\Omega} }$ is large, and to show that some dependence on this parameter is necessary.  We will not be able to obtain tight bounds on the dependence on $\lambda$, but we will be able to obtain upper and lower bounds that have similar qualitative dependence on $\lambda$ in some parameter regimes.  We leave it as an interesting open problem to understand the ``correct'' dependence on $\lambda$.

While this seems like a difficult challenge in general, we are able to make progress when $\ind{\Omega}$ happens to be the adjacency matrix of a connected graph.  In this case, $\lambda$ is directly related to the spectral gap of the underlying graph, and there are many tools available to study it.  Thus, our results below apply to this special case.

\subsection{Upper bound}
In this section we specialize our upper bound, Theorem \ref{lem:upperrankk}, to the case where $\lambda$ is large.

\begin{theorem}[Upper bound for rank-$r$ matrices in terms of $\lambda_1$ and $\lambda_2$]\label{thm:gapupper}
Let $\Omega \subseteq [d] \times [d]$ be such that $\ind{\Omega}$ is the adjacency matrix of a connected, undirected graph on $d$ vertices. 
Let $\bv_1, \bv_2$ be the eigenvectors of $\ind{\Omega}$ corresponding to the top two largest eigenvalues (by magntitude), $\lambda_1, \lambda_2$.
Suppose that there is some constant $C$ so that
\[ \max_i |(v_1)_i | \leq C \min_i |(v_1)_i| \qquad \max_i| (v_2)_i | \leq C \min_i |(v_2)_i | \]
Let $\sigma > 0$ and let $\bW$ denote the best rank-1 approximation to $\ind{\Omega}$ .

Suppose that $\bM \in \RR^{d \times d}$ has rank $r$.  Suppose that $Z_{ij} \sim \gN(0,\sigma^2)$ and let 
\[ \hat{ \bM } = \bW^{(-1/2)} \had \argmin_{ \text{rank}(\bX) = r } \norm{ \bX - \bW^{(-1/2)} \had (\bM_\Omega + \bZ_\Omega) }. \]
Then with probability at least $1 - 1/2d$ over the choice of $\bZ$,
\[ \frac{1}{\fronorm{ \bW^{(1/2)} }}\fronorm{ \bW^{(1/2)} \had ( \bM - \hat{ \bM } ) } \leq c \inparen{ r \inparen{ \frac{ \lambda_2 }{\lambda_1 } } + \sigma \sqrt{ \frac{ r \log(d) } {\lambda_1 } } }, \]
where $c$ is a constant that depends only on $C$.
\end{theorem}

\begin{proof}
To apply Theorem~\ref{lem:upperrankk}, we must compute $\lambda, \mu$ and $\fronorm{ \bW^{(1/2)} }$ in terms of $\lambda_1$.
We can explicitly write $\bW = \lambda_1 \bv_1 \bv_1^T$, since $\lambda_1$ is the largest eigenvalue of $\ind{\Omega}$.
Of course, since $\ind{\Omega}$ defines the adjacency matrix of a connected graph, $\bv_1$ has strictly positive entries
(by the Perron-Frobenius Theorem) and hence $W_{ij} > 0$ for all $i, j$. 
Let $\bar{W} = \min_{ij} W_{ij}$, so that, using our assumption $W_{ij} \in [\bar{W}, C \bar{W} ]$ for all $i,j$.

Since $\bar{W}$ is rank-1, 
\[ \lambda_1 = \opnorm{ \bW } = \fronorm{ \bW} = \sqrt{ \sum_{ij} W_{ij}^2} \in [ \bar{W} d, C \bar{W} d].  \]
Rearranging, this means
\[ \bar{W} \in \left[ \frac{ \lambda_1 }{Cd }, \frac{ \lambda_1 }{d } \right]. \]
Then we compute
\[ \lambda = \opnorm{ \bW^{(1/2)} - \bW^{(-1/2)} \had \ind{\Omega} } \leq \frac{1}{ \sqrt{ \bar{W}} } |\lambda_2| \leq \sqrt{ \frac{ Cd }{ \lambda_1} }|\lambda_2| \]
and 
\[ \mu^2 = \max\inset{ \max_i \sum_j \frac{ \Omega_{ij} }{W_{ij}} , \max_j \sum_i \frac{ \Omega_{ij}}{W_{ij}} } 
\leq \frac{1}{\bar{W}} \max_i \| \br_i\|_0 \leq \frac{C d}{\lambda_1} \max_i \| \br_i\|_0 \]
where $\br_i$ is the $i$'th row of $\ind{\Omega}$, and $\|\cdot\|_0$ denotes the number of nonzero entries.
Now we have, for all $i$, 
\[ \ip{ \br_i }{ \bv_1 } = \lambda_1 (v_1)_i \in [ \lambda_1 \bar{v}, C \lambda_1 \bar{v} ], \]
where $\bar{v} = \min_i (v_1)_i > 0$.  
We also have
\[ \ip{ \br_i }{ \bv_1 } = \sum_{j=1}^d (r_i)_j (v_1)_j \in [ \bar{v} \| \br_i \|_0, C \bar{v} \| \br_i \|_0 ], \]
and together these imply that
\[ \frac{ \lambda_1 }{C } \leq \| \br_i \|_0 \leq C \lambda_1 \]
for all $i$.
Thus, we can simplify the bound on $\mu^2$ to 
\[ \mu^2 \leq \frac{ C d }{ \lambda_1 } \max_i \|\br_i\|_0 \leq  C^2 d.  \]
Finally, we bound (as $C \geq 1$)
\[ \fronorm{ \bW^{(1/2)} } = \sqrt{ \sum_{i,j} W_{ij} } \in \left[ d \sqrt{ \bar{W} }, Cd \sqrt{ \bar{W} } \right] \subseteq \left[ \frac{ 1}{C} \sqrt{  \lambda_1 d} , C \sqrt{ \lambda_1 d } \right]. \]
Plugging all of these bounds into Theorem~\ref{lem:upperrankk} proves the theorem (and $c = 4\sqrt{2} C^2$ suffices). 
\end{proof}

\subsection{Lower bound}

\begin{theorem}\label{thm:gaplower}
Let $\Omega \subseteq [d] \times [d]$ be such that $\ind{\Omega}$ is symmetric and corresponds to the adjacency matrix of a connected undirected graph on $d$ vertices.
Let $\bv_1$ and $\bv_2$ be the first and second eigenvectors of $\ind{\Omega}$, (normalized so that $\| \bv_1 \|_2 = \| \bv_2 \|_2 = 1$) with corresponding eigenvalues $\lambda_1$ and $\lambda_2$.  
Suppose that 
\[ \max_i | (v_1)_i | \leq C \min_i |(v_1)_i| \qquad \max_i |(v_2)_i | \leq C \min_i |(v_2)_i|. \]
Let $\bW = \lambda_1 \bv_1 \bv_1^T$ be the best rank-1 approximation to $\ind{\Omega}$.

Let $\sigma > 0$, let $0 < r < d^{1/3}/\log^{2/3}(d)$, and let $K \subset \R^{d \times d}$ be the cone of rank $r$ matrices.  For any algorithm $\cA: \R^\Omega \to \R^{d \times d}$ that takes as input $\bX_\Omega + \bZ_\Omega$ and outputs a guess $\hat{\bX}$ for $\bX$, for $\bX \in K \cap \beta B_\infty$ and $Z_{ij} \sim \gN(0,\sigma^2)$, there is some $\bM \in K \cap \beta B_\infty$ so that
\[
\frac{1}{ \fronorm{ \bW^{(1/2)} } }
 \fronorm{ \bW^{(1/2)} \had ( \cA( \bM_\Omega + \bZ_\Omega )  - \bM ) } 
\geq c \min \inset{ \sqrt{ \frac{ 1 }{ \lambda_1 - \lambda_2 } } \cdot \sigma \sqrt{ r } , \ \frac{ \beta } { \sqrt{ r \log(d) } } } 
\]
\end{theorem}
\begin{proof}
By the Perron-Frobenius theorem, $\bv_1 \succ \mathbf{0}$, aka, it has all strictly positive entries.
Since $\bv_1 \perp \bv_2$, $\bv_2$ must have some negative entries.  Without loss of generality, suppose that the coordinates are ordered so that the entries of $\bv_2$ are decreasing, and write
\[ \bv_2 = ( \bh, -\bg ) \]
where $\bh, \bg \succeq \mathbf{0}$, $\bh \in \R^s$  and $\bg \in \R^t$.  (This defines $s,t$ so that $s + t = d$).  Write $\bv_1 = (\ba, \bb)$ for $\ba \in \R^s$, $\bb \in \R^t$ according to the same partition of coordinates, and write
\[ \ind{\Omega} = \begin{pmatrix} \bB_0 & \bB \\ \bB^T & \bB_1 \end{pmatrix} \]
so that $\bB_0 \in \R^{s \times s}$, $\bB_1 \in \R^{t\times t}$ and $\bB \in \R^{s \times t}$.
Notice that by orthogonality, 
\[ 0 = \ip{ \bv_1 }{\bv_2} = \ip{ \ba }{\bh} - \ip{\bb}{ \bg } \]
and so 
\[ \ip{ \ba}{\bh} = \ip{ \bb}{\bg }. \]
Let
\[ \alpha := 2 \ip{ \ba }{\bh } = 2 \ip{ \bb }{\bg} = \ip{ \ba }{\bh } + \ip{ \bb }{ \bg } . \]
\begin{claim}
\[
\bh^T \bB \bg \leq \frac{ \alpha ( \lambda_1 - \lambda_2 ) }{ 4 }. \]
\end{claim}
\begin{proof}
Let $\bx = (\bh, \bg)$, so that $\bx \succeq \mathbf{0}$.  Then 
\[ \bx^T \ind{\Omega} \bx = \bh^T \bB_0 \bh + \bg^T \bB_1 \bg + 2 \bh^T \bB \bg \]
and
\[ \lambda_2 = \bv_2^T \ind{\Omega} \bv_2 = \bh^T \bB_0 \bh + \bg^T \bB_1 \bg - 2 \bh^T \bB \bg, \]
so
\[ \bx^T \ind{\Omega} \bx = \lambda_2 + 4 \bh^T \bB \bg. \]
Observing that $\alpha = \ip{ \bx }{ \bv_1}$, we have
\[ \bx = \alpha \bv_1 + \inparen{\sqrt{ 1 - \alpha^2 }} \bz, \]
for some vector $\bz$ that satisfies $\bz \perp \bv_1$.  Then 
\begin{align*}
\lambda_2 + 4 \bh^T \bB \bg &= \bx^T \ind{\Omega} \bx \\
&= \alpha^2 \bv_1^T \ind{\Omega} \bv_1 + (1 - \alpha^2) \bz^T \ind{\Omega} \bz + 2 \alpha \sqrt{ 1 - \alpha^2 } \bv_1^T \ind{\Omega} \bz \\
&= \alpha^2 \bv_1^T \ind{\Omega} \bv_1 + (1 - \alpha^2) \bz^T \ind{\Omega} \bz \\
&\leq \alpha^2 \lambda_1 + (1 - \alpha^2) \lambda_2,
\end{align*}
from which we conclude that
\[ 4 \bh^T \bB \bg \leq \alpha^2 (\lambda_1 - \lambda_2), \]
as desired.  Above, we have used the fact that $\bv_1^T \ind{\Omega} \bz = \lambda_1 \ip{ \bv_1 }{\bz} = 0$ since $\bv_1 \perp \bz$.
\end{proof}
Now define
\[ \bW = \lambda_1 \bv_1 \bv_1^T \]
to be the best rank-1 approximation to $\ind{\Omega}$.  
Choose
\[ \bA = \frac{1}{ \sqrt{ \infnorm{ \bh } \infnorm{ \bg } } } \inparen{ ( \bh, \mathbf{0}) ( \mathbf{0}, \bg )^T }^{(1/2)} \]
and 
\[ \bH = \bW^{(1/2)} = \sqrt{ \lambda_1 } \inparen{ ( \ba, \bb ) ( \ba , \bb)^T )}^{(1/2)}. \]
Then we may compute
\[ \fronorm{ \bH \had \bA }^2 = \frac{ 1 }{ \infnorm{ \bh }{ \infnorm{ \bg } } } \lambda_1 \ip{ \ba }{ \bh } \ip{ \bb }{\bg } 
= \frac{ \lambda_1 \alpha^2 }{ 4 \infnorm{ \bh } \infnorm{ \bg } } \]
and
\[ \fronorm{ \bA_\Omega }^2 = \frac{1}{ \infnorm{ \bh } \infnorm{ \bg }} \sum_{i=1}^s \sum_{j=1}^t B_{ij} h_i g_j = \frac{ \bh^T \bB \bg }{ \infnorm{ \bh } \infnorm{ \bg } } \leq \frac{ \alpha^2 ( \lambda_1 - \lambda_2 ) }{ 4 \infnorm{ \bh } \infnorm{ \bg } } \]
using the claim.
Now we apply Lemma~\ref{lem:net} with this choice of $\bA$.  In the language of that lemma, we have
\[ \bz = \inparen{ \frac{ \lambda_1 }{ \infnorm{ \bh } \infnorm{ \bg } } } ( \bh \had \ba, \mathbf{0} ) \qquad \bv = \inparen{ \frac{ \lambda_1 }{ \infnorm{ \bh } \infnorm{ \bg } } } ( \mathbf{0}, \bg \had \bb  ). \]
 
By our assumption, there is some constant $C$ so that $\max_i|h_i| \leq C \min_i|h_i|$, and the same for $\bg, \ba, \bb$.  
Thus, as in the proof of Lemma~\ref{lem:flatLB}, Lemma~\ref{lem:net} guarantees a net $\cX$ so that 
\[ |\cX| \geq 2e \exp( C'rd ), \]
for some constant $C'$ (which depends on $C$) and so that, for some constant $c$,
for all $\bX \in \cX$ we have
\[ \fronorm{ \bX_\Omega } \leq \frac{ c \alpha \sqrt{r  (\lambda_1 - \lambda_2) } }{ \sqrt{ \infnorm{ \bh }\infnorm{ \bg } } }, \]
for all $\bX \neq \bX' \in \cX$ we have
\[ \fronorm{ \bH \had ( \bX - \bX' ) } \geq \frac{ \sqrt{r \lambda_1} \alpha }{\sqrt{  \infnorm{ \bh } \infnorm{ \bg }} }, \]
and finally for all $\bX \in \cX$ we have
\[ \infnorm{ \bX } \leq c \sqrt{ r \log(d) } . \]
Now we want to apply Lemma~\ref{lem:fano}, and we choose
\[ \kappa = \min \inset{ c'' \inparen{ \frac{ \sigma }{ \alpha }} \sqrt{ \frac{ rd  \infnorm{ \bh } \infnorm{ \bg } }{ \lambda_1 - \lambda_2 } } , \frac{ \beta }{ c\sqrt{ r \log(d) } } }. \]
for some constant $c''$ to be chosen below.
Observe that
\begin{align*}
\frac{ \sigma \sqrt{ \log |\cX| } }{ 4 \max_{ \bX \in \cX } \fronorm{ \bX } } 
&= \frac{ \sigma \sqrt{ C'rd } \sqrt{ \infnorm{ \bh} \infnorm{ \bg } }}{ 2 \alpha \sqrt{ \lambda_1 - \lambda_2 } } 
\geq \kappa
\end{align*}
for an appropriate choice of $c'' = \sqrt{C'}/2$, so this is a legitimate choice of $\kappa$ for Lemma~\ref{lem:fano}.  We conclude that for any algorithm $\cA$ that works on $K \cap \beta B_\infty$, there is some matrix $\bM \in K \cap \beta B_\infty$ so that
\begin{align*}
 \fronorm{ \bW^{(1/2)} \had ( \cA( \bM_\Omega + \bZ_\Omega )  - \bM ) } &\geq \frac{ \kappa }{2 } \min_{ \bX \neq \bX' \in \cX } \fronorm{ \bH \had (\bX - \bX') } \\
&\geq  \frac{1}{2} \min \inset{ c'' \inparen{ \frac{ \sigma }{ \alpha }}  \sqrt{ rd \frac{  \infnorm{ \bh } \infnorm{ \bg } }{ \lambda_1 - \lambda_2 } } , \frac{ \beta }{ c\sqrt{ r \log(d) } } } \cdot \sqrt{\frac{ \lambda_1 \alpha^2 }{ \infnorm{ \bh } \infnorm{ \bg } } }\\
&\geq c''' \min \inset{ \sqrt{ \frac{ \lambda_1 }{ \lambda_1 - \lambda_2 } } \sigma \sqrt{ rd} , \beta \sqrt{ \frac{ \lambda_1 \alpha^2 }{ \infnorm{ \bh }\infnorm{ \bg } r\log(d) } } } 
\end{align*}
for some constant $c'''$.

\begin{claim}
Under the assumptions on $\bv_1, \bv_2$, 
we have
\[ \frac{1}{C^2 d } \leq \infnorm{\bh} \infnorm{ \bg } \leq \frac{C^2}{d}, \]
and
\[ \alpha \geq \frac{1}{C^2}. \]
\end{claim}
\begin{proof}
First, we observe that since $\bv_1, \bv_2$ have $\ell_2$-norm $1$ and have entries that are all about the same magnitude, up to a factor of $C$, we must have every entry of $\bv_1$ and $\bv_2$ in the interval $\left[ \frac{1}{C \sqrt{d} }, \frac{C}{\sqrt{d}} \right]$.
(Indeed, if one entry of $\bv_1$ is larger than $C/\sqrt{d}$, then all entries are larger than $1/\sqrt{d}$, which contradicts the requirement on the $2$-norm, and similarly if the smallest entry is smaller than $1/(C\sqrt{d})$; the same holds for $\bv_2$).   This immediately implies the claim about $\infnorm{ \bh } \infnorm{ \bg }$ since $\bv_2 = ( \bh, - \bg )$.

Next, we recall that 
\[ \alpha = \ip{ \bh }{ \ba } + \ip{ \bg }{ \bb } =  2 \ip{ \bh }{ \ba } = 2 \ip{ \bg }{ \bb}, \] and
using the observation above about the size of the entries of $\bv_1, \bv_2$, we have
\[ \alpha \geq \frac{s}{C^2 d} + \frac{t}{ C^2 d } = \frac{ 1}{C^2}. \]

\end{proof}
Using the claim, we bound $\alpha, \infnorm{ \bh } \infnorm{ \bg }$ in the second term above, and conclude
that there is some $\bM$ so that
\begin{align*}
 \fronorm{ \bW^{(1/2)} \had ( \cA( \bM_\Omega + \bZ_\Omega )  - \bM ) } 
&\geq c'''' \min \inset{ \sqrt{ \frac{ \lambda_1 }{ \lambda_1 - \lambda_2 } } \cdot \sigma \sqrt{ rd} , \ \beta \sqrt{ \frac{ d \lambda_1 }{  r\log(d) } } } 
\end{align*}
for some other constant $c''''$.  This completes the proof, after normalizing by
\[ \fronorm{ \bW^{(1/2)} } = \sqrt{ \lambda_1 \onenorm{ \bv_1 }^2 } \leq C \sqrt{d \lambda_1 }. \]
\end{proof}

\section{Experiments}\label{sec:experiments}
In this section, we illustrate the results of numerical experiments for our
debiased projection method,
\[
\hat \bM_{\rm debias} := \bW^{(-1/2)} \circ \argmin_{\rank(\bX) = r} \| \bX - \bW^{(-1/2)} \circ \bM \|_F.
\]
In this section, we refer to this algorithm as a \emph{debiased}, low-rank projection. This is in contrast to a
standard (see section 2.5 of \cite{plan2016high}) low-rank projection,
\[
\hat \bM_{\rm standard} := \argmin_{\rank(\bY) = r} \left\| \bY - p^{-1} \bM \right\|_F.
\]
Above, $p := |\Omega|/d_1d_2$. 
We report the performance of these two procedures for various synthetic and data-derived sampling patterns.

\subsection{Data-derived sampling patterns}  \let\tilde\widetilde

In our first experiments, we use sampling patterns taken from the Jester joke corpus \cite{eigentaste2001goldberg}, and the
Movielens 100k dataset \cite{harper2016movielens}.
In the first dataset, users rate jokes; in the second, users rate movies.
We take $d_1$ to be the number of users enrolled in the dataset, and
$d_2$ to be the number of rated items (\emph{e.g.}, jokes or films).
For the Jester dataset, we have $d_1 = 73, 421$ and $d_2 = 100$,
For the Movielens dataset, we have $d_1 = 997$ and $d_2 = 538$. 
We take $\Omega \subseteq [d_1] \times [d_2]$ to be the observed indices in each
dataset. In particular, $(i, j) \in \Omega$ whenever user $i$ rates item $j$.
With $\ones_\Omega = \left(\mathbf{1}_{(i, j) \in \Omega}\right)_{\substack{1 \leq i \leq d_1 \\ 1 \leq j \leq d_2}}$,
in the following experiments we take $\bW$ to be the best rank-1 approximation to $\ones_\Omega$,
$\bW := \argmin_{\rank(U) = 1} \|\ones_\Omega - U\|$.
In the examples we present, $\bW_{ij} \geq 0$ for all $i, j$. 

To produce Figures \ref{fig:real-data} and \ref{fig:real-data-2} below, we generate synthetic low-rank matrices, and use the sampling patterns from the real-life data described above. 
We consider ranks $r$ between $1$ and $10$.
For $N = 50$ trials, we construct random rank $r$ matrices $\bX_1, \dots, \bX_N \in \R^{d_1 \times d_2}$, with independent standard normal entries.
For each matrix $\bX_i$, we average over $T = 25$ tests, testing our algorithm versus the the standard projection algorithm (\emph{e.g.}, truncated SVD)
on $\bY_{i, 1}, \dots, \bY_{i, T}$, where $\bY_{i, j} = \ones_\Omega \circ (\bX_i + \bZ_{i, j})$, where $\bZ_{i, j} \in \R^{d_1 \times d_2}$ has independent
standard normal entries. We measure error using the weighted Frobenius norm. In particular, with $\hat \bX$ an estimate of $\bX$, we report
\[
\frac{\|\bW^{(1/2)} \circ (\bX - \hat \bX)\|_F}{\|\bW^{(1/2)}\|_F} \quad \text{and} \quad
\frac{\|\bX - \hat \bX\|_F}{\sqrt{d_1d_2}}.
\]
In the sequel, we refer to these as the \emph{weighted error} and \emph{unweighted error}, respectively.

\begin{figure}[h!t]
\centering
\begin{overpic}[width=0.45\linewidth]{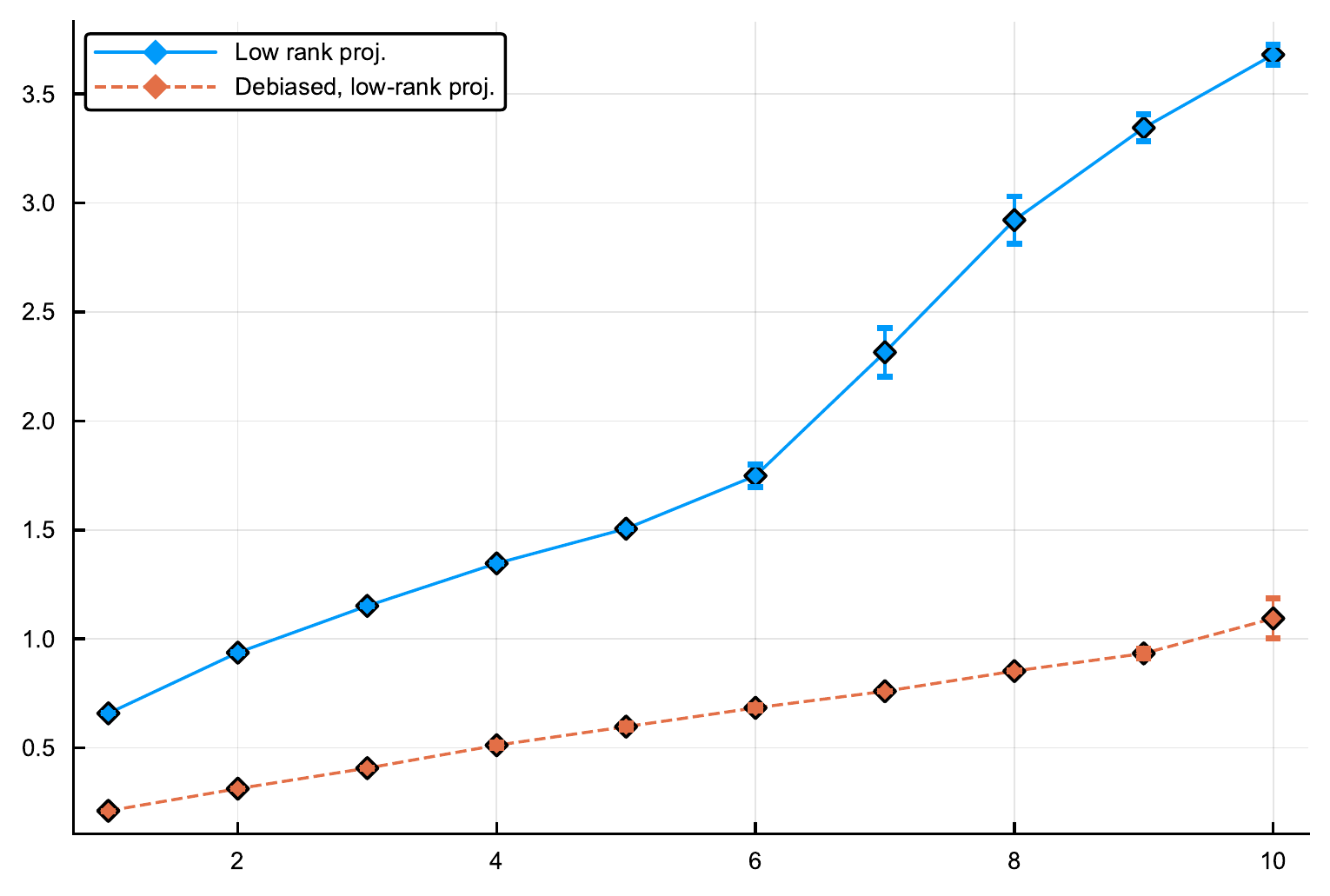}
\put(45, -3){\scalebox{0.75}{rank}}
\put(-4, 28){\scalebox{0.75}{\rotatebox{90}{{weighted error}}}}
\end{overpic}
\hfill
\begin{overpic}[width=0.45\linewidth]{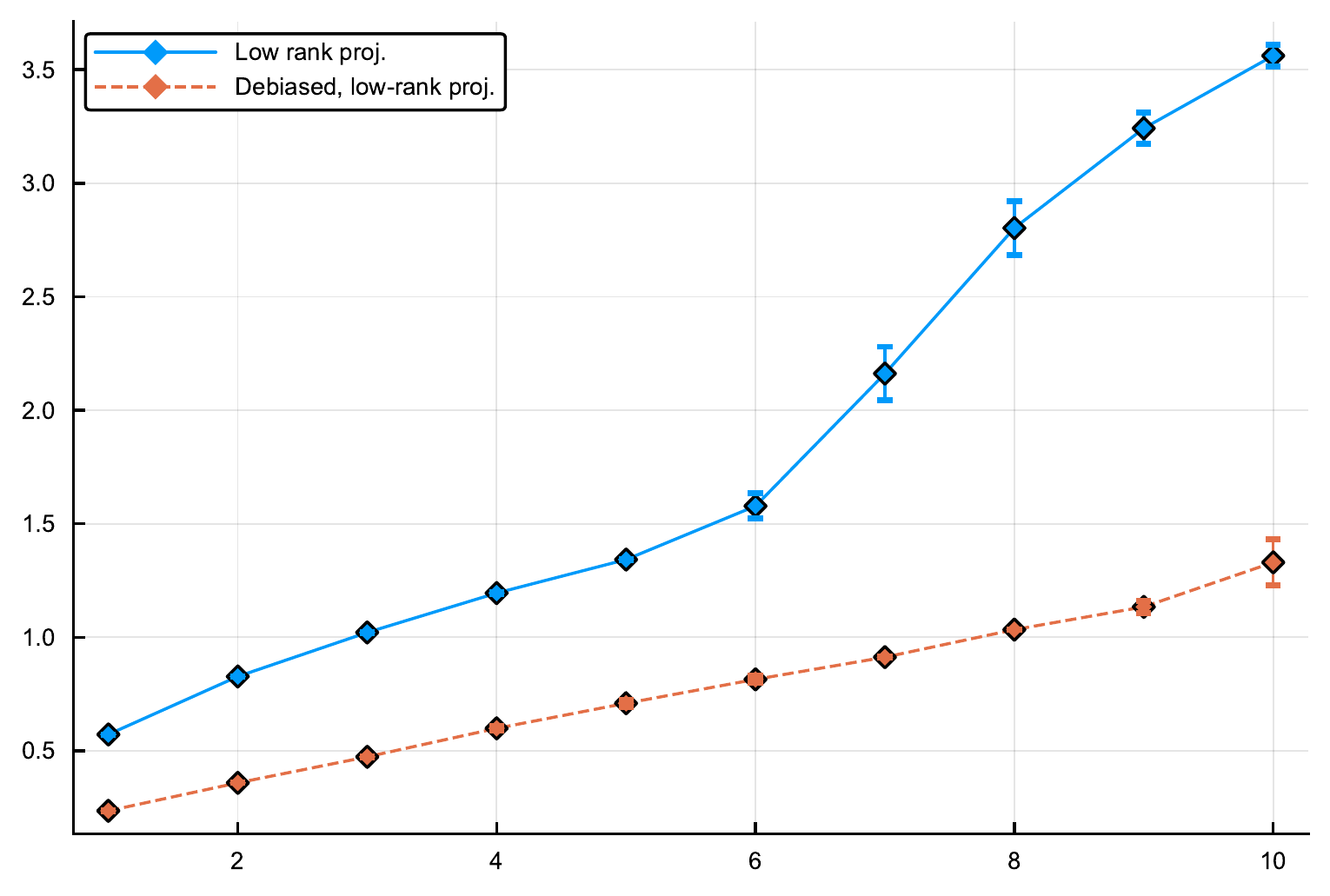}
\put(45, -3){\scalebox{0.75}{rank}}
\put(-4, 28){\scalebox{0.75}{\rotatebox{90}{{unweighted error}}}}
\end{overpic}
\vspace{1mm}
\caption{Average errors of our debiased algorithm, versus standard low-projection procedure
  on Jester sampling pattern and standard Gaussian data. Error bars denote standard deviations over 30 samples of the
  noise; each line (and error bar) is averaged over 30 trials.}
\label{fig:real-data}
\end{figure}

\begin{figure}[h!t]
\centering
\begin{overpic}[width=0.45\linewidth]{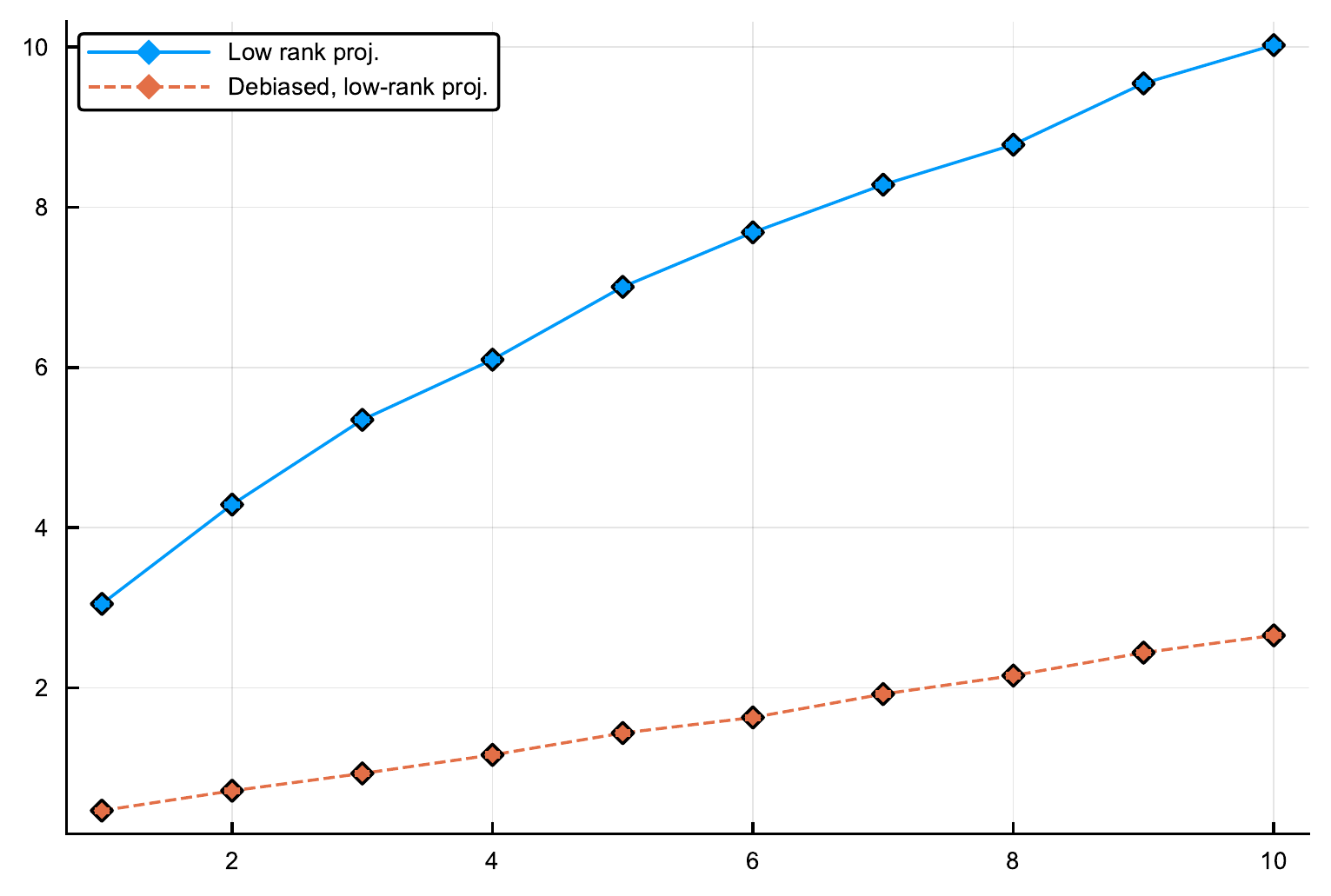}
\put(45, -3){\scalebox{0.75}{rank}}
\put(-4, 28){\scalebox{0.75}{\rotatebox{90}{{weighted error}}}}
\end{overpic}
\hfill
\begin{overpic}[width=0.45\linewidth]{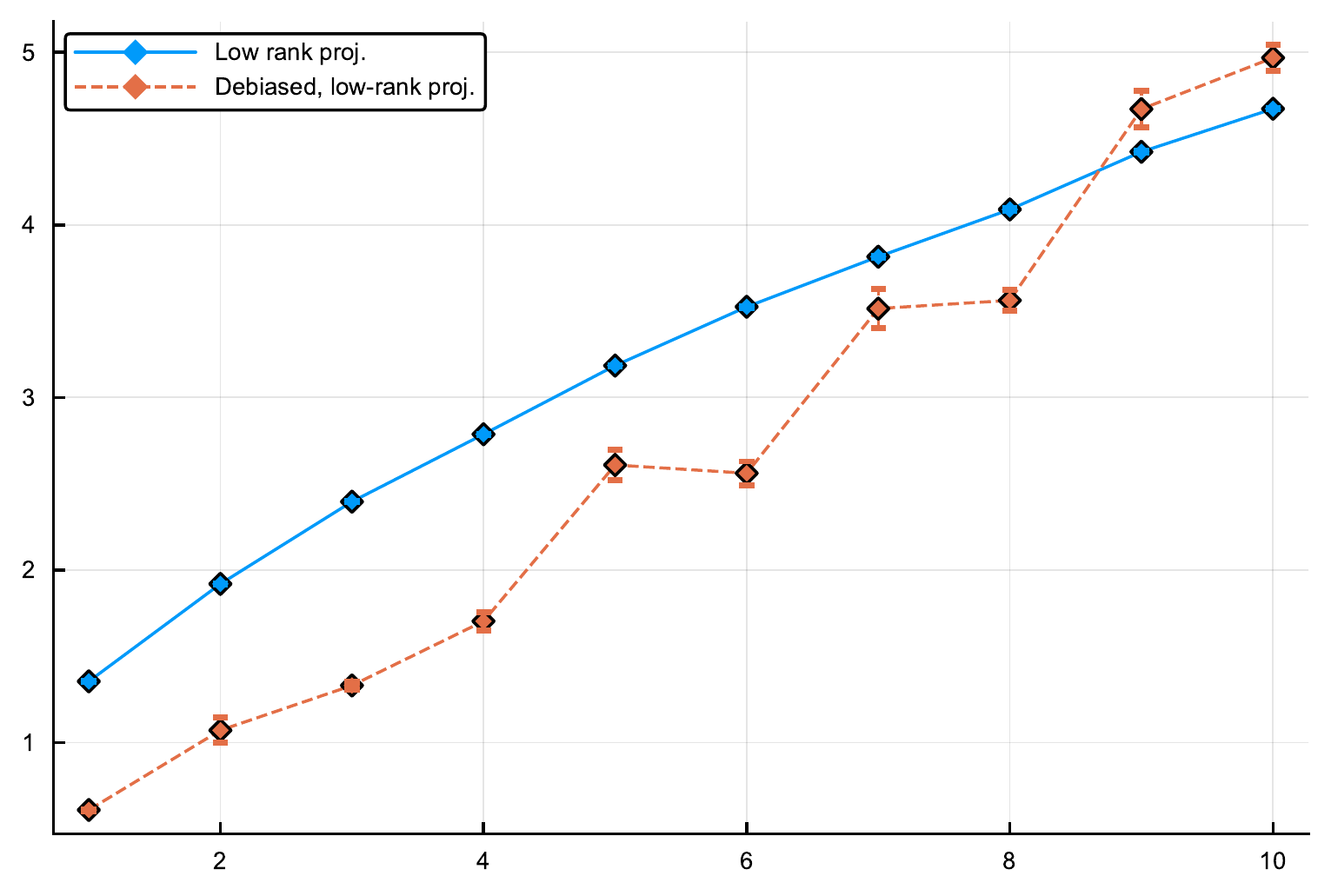}
\put(45, -3){\scalebox{0.75}{rank}}
\put(-4, 28){\scalebox{0.75}{\rotatebox{90}{{unweighted error}}}}
\end{overpic}
\vspace{1mm}
\caption{Average errors of our debiased algorithm, versus standard low-rank projection procedure
  on Movielens sampling pattern and standard Gaussian data. Error bars denote standard deviations over 30 samples of the
  noise; each line (and error bar) is averaged over 30 trials.}
\label{fig:real-data-2}
\end{figure}

We remark that although we only provide guarantees on the performance in the weighted Frobenius norm, our procedures exhibit
good empirical performance (relative to standard projection procedures) even in the usual Frobenius norm. 

\subsection{Synthetic sampling patterns}
For both of the experiments below, we focus on various synthetic constructions of $\Omega$
in order to demonstrate how the performance of debiased projection depends on the parameters of the sampling pattern.

For $d$ and $r$ specified below, we always construct random data as follows. For $N = 15$,
we pick matrices $\bX_1, \dots, \bX_N \in \R^{d \times d}$ with for $n = 1, \dots, N$,
\[
\bX_n = \bU_n \bV_n^T, \quad \bU_n, \bV_n \in \R^{d \times r}, \quad (\bU_n)_{ij}, (\bV_n)_{ij} \overset{iid}{\sim} \mathcal{N}(0, 1).
\]

We describe two experiments below.  The first one is motivated by the first case study, when $\Omega \sim \bW$, and the second is motivated by the second case study when the spectral gap is large.

\subsubsection{Sampling $\Omega \sim \bW$}
In the following experiment, we simulate our first case study, sampling $\Omega \sim \bW$ for a rank-$1$ matrix $\bW$.
For simplicity, we take the weight matrix $\bW = \bw \bw^T$ to be symmetric. We also take $d = 1000$ and $r = 10$.

We choose several different $\bW$'s with different levels of ``flatness,'' to show how the performance of our algorithm depends on the flatness of $\bW$.
More precisely,
let $m \in [4 d \log d, d^2/4]$ and $y \in [\sqrt{2/d}, \sqrt{\log d/d}]$.
For each $m$ and $y$, we construct sampling patterns, $\Omega_1, \dots, \Omega_T$, with $T = 15$ in the following manner.
We select $\bW$ with $\bw$ given by 
\[
\bw = \left(f(y, m, d) \ones_{d/ 2}, y \ones_{d/2}\right) \qquad \mbox{where} \qquad f(y, m, d) = \frac{2\sqrt{m}}{d} - y.
\]
Above, $f(y, m, d)$ is chosen such that $\|\bw\|_1 = \sqrt{m}$, and hence $\E{|\Omega|} = m$, when $\Omega \sim \bW$.
Thus, a larger value of $y$ corresponds to a ``flatter'' matrix $\bW$.

With such choices of $\bw$, $\bW_{ij} \in (0, 1]$, and we draw $\Omega_t \sim \bW$ for $t = 1, \dots T$.
  For each $t$, we run the standard truncated SVD algorithm and our debiased projection procedure on
  $Y_{n, t} = \mathbf{1}_{\Omega_t} \had (\bX_n + \bZ_{n, t})$, for $n = 1, \dots, N$ and $t = 1, \dots, T$.
  Here, $(\bZ_{n, t})_{ij} \sim \mathcal{N}(0, 1)$ for $1 \leq i, j \leq d$. We repeat
  this experiment for various $m$ and $y$ in the intervals above. 
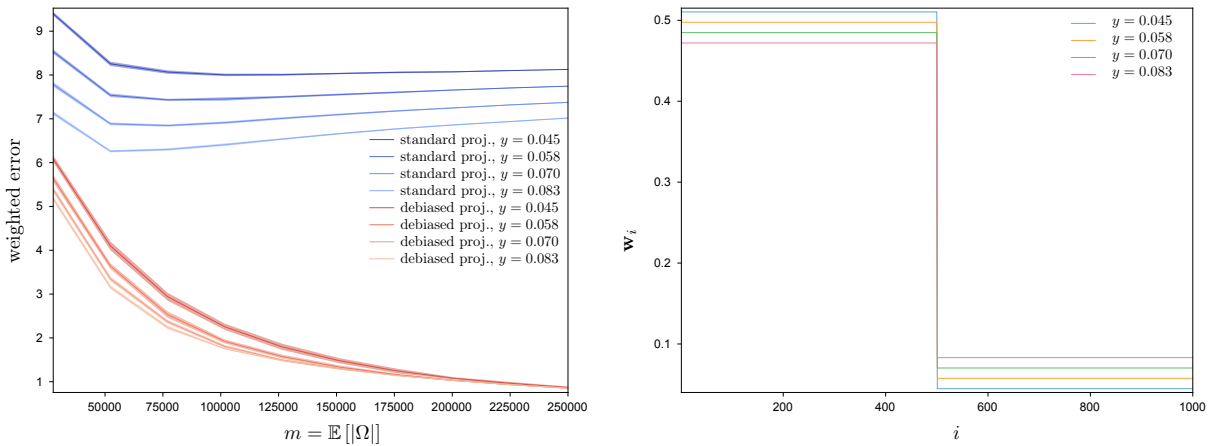
\begin{figure}[h!t]
\centering
\begin{overpic}[width=0.475\linewidth]{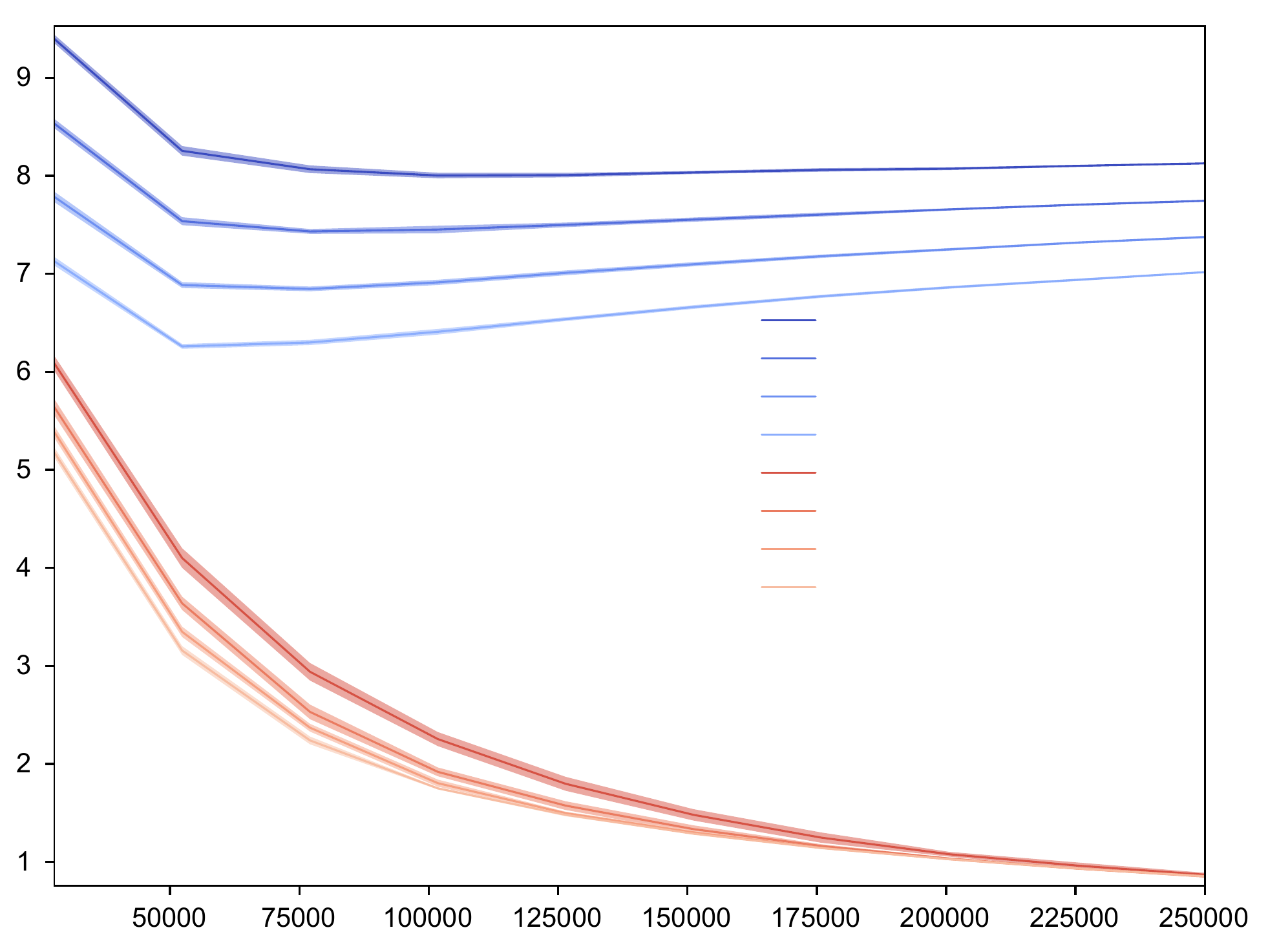}
\put(45, -3){\scalebox{0.75}{$m = \E{|\Omega|}$}}
\put(-4, 28){\scalebox{0.75}{\rotatebox{90}{{weighted error}}}}
\put(65.5, 49){\scalebox{0.55}{standard proj., $y = 0.045$}}
\put(65.5, 46){\scalebox{0.55}{standard proj., $y = 0.058$}}
\put(65.5, 43){\scalebox{0.55}{standard proj., $y = 0.070$}}
\put(65.5, 40){\scalebox{0.55}{standard proj., $y = 0.083$}}
\put(65.5, 37){\scalebox{0.55}{debiased proj., $y = 0.045$}}
\put(65.5, 34){\scalebox{0.55}{debiased proj., $y = 0.058$}}
\put(65.5, 31){\scalebox{0.55}{debiased proj., $y = 0.070$}}
\put(65.5, 28){\scalebox{0.55}{debiased proj., $y = 0.083$}}
\end{overpic}
\hfill
\begin{overpic}[width=0.475\linewidth]{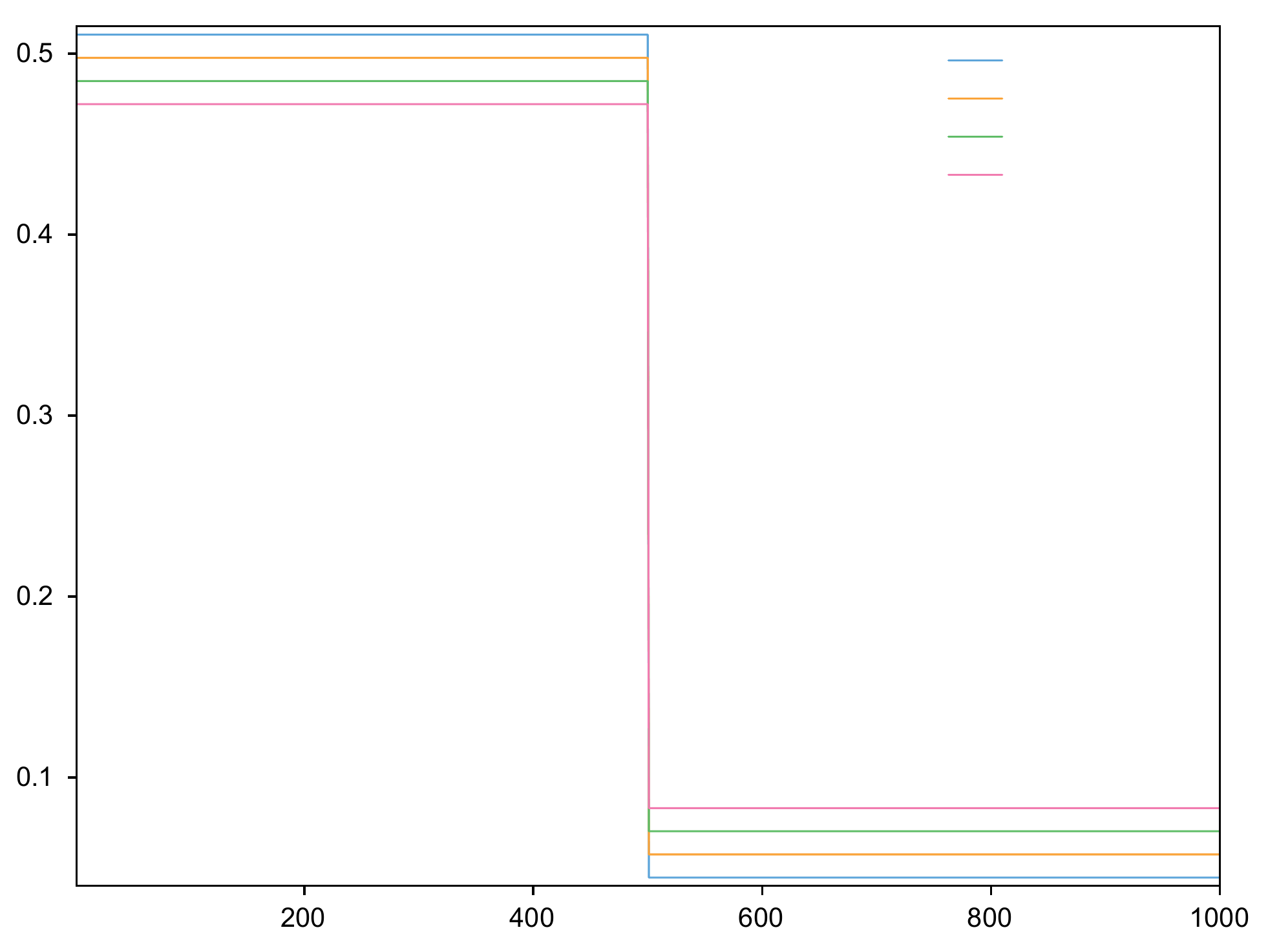}
\put(54, -3){\scalebox{0.75}{$i$}}
\put(-4, 30){\scalebox{0.75}{\rotatebox{90}{$\bw_i$}}}
\put(82, 70){\scalebox{0.55}{$y = 0.045$}}
\put(82, 67){\scalebox{0.55}{$y = 0.058$}}
\put(82, 64){\scalebox{0.55}{$y = 0.070$}}
\put(82, 61){\scalebox{0.55}{$y = 0.083$}}
\end{overpic}
\vspace{1mm}
\caption{(Left) Weighted error for various $\Omega \sim \bW$ as described above. Shading indicates a single standard deviation, across the
  draws $\Omega_t \sim \bW$. (Right) Entries of $\bw$ for various choices of $y$ and $m = 77,045$.}
\label{fig:sampleW}
\end{figure}
The range of $y$ is chosen to ensure that the plots of weighted error above are over the same range of $m$, while still respecting
the constraints on $\bW$ (\emph{i.e.}\, that it have nonnegative entries in $[1/\sqrt{d}, 1]$).
 
Figure \ref{fig:sampleW} indicates
that under the experimental conditions given above, the SVD procedure has worse performance as $\bw$ becomes flatter, while as expected
our debiased procedure has improved performance as $\bw$ becomes flatter.  The debiased projection algorithm out-performed the standard projection in all cases.

\subsubsection{Dependence on spectral gap}
The following experiment demonstrates how the performance of the debiased projection algorithm depends on the spectral gap of $\ones_\Omega$. 

To construct sampling patterns of various spectral gaps, we consider graph products on $k = 50$ vertices.
For regularities $\rho \in [2, 20]$, we construct two graphs $G_\rho, \tilde G_\rho$ on $k$ vertices that are $\rho$-regular.

The graph $G_\rho$ is constructed such that each vertex $v \in \{0, \dots, k-1\}$ is adjacent to the vertices $v' \equiv v + t \mod k$,
$t = 0, 1, \dots, \rho$.  Notice that $G_\rho$ is the same graph as is considered in Example~\ref{ex:fatcycle}, and the spectral gap is quite small. 
The second graph will be a random $\rho$-regular graph, so the spectral graph is with high probability large.  More precisely,
for $T =15$ and trials $t = 1, \dots, T$, the graph $\tilde G^{(t)}_\rho$ is constructed as a random $\rho$-regular graph \cite{bollobas2001random}.

In order to consider graphs with a range of $\lambda$, we consider three distributions on graphs: $G_\rho \otimes G_\rho$ (which has a small spectral gap), $G_\rho \otimes \tilde{G}^{(t)}_\rho$ (which has an intermediate spectral gap) and $\tilde{G}^{(t)}_\rho \otimes \tilde{G}^{(t)}_\rho$ (which has a large spectral gap).  

For $t = 1, \ldots, T$, we let $\Omega_t$ be the sampling pattern induced by each of these three graphs, and we draw observations $\bY_{n, t} = \mathbf{1}_{\Omega_t} \had (\bX_{n, t} + \bZ_{n, t})$ as in the previous set of experiments.
As before, we
 carry out both our debiased projection and truncated SVD on the $\bY_t$. In these experiments, we take the rank of the data to be $r = 10$. 
The results are shown in Figure~\ref{fig:spectralgap}.
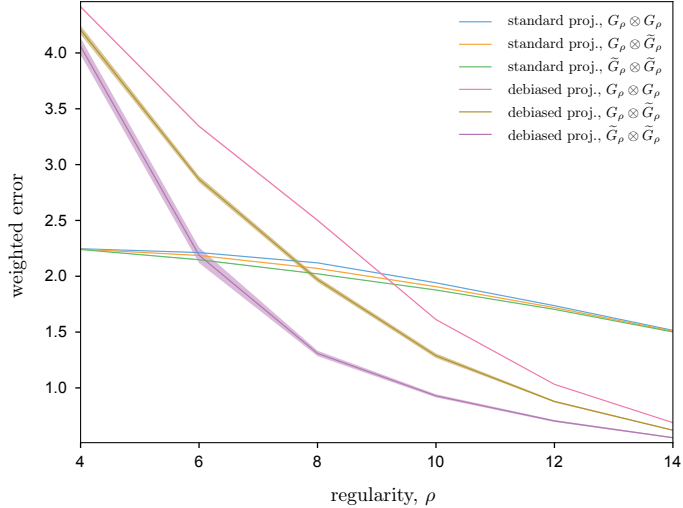
\begin{figure}[h!t]
\centering
\begin{overpic}[width=0.55\linewidth]{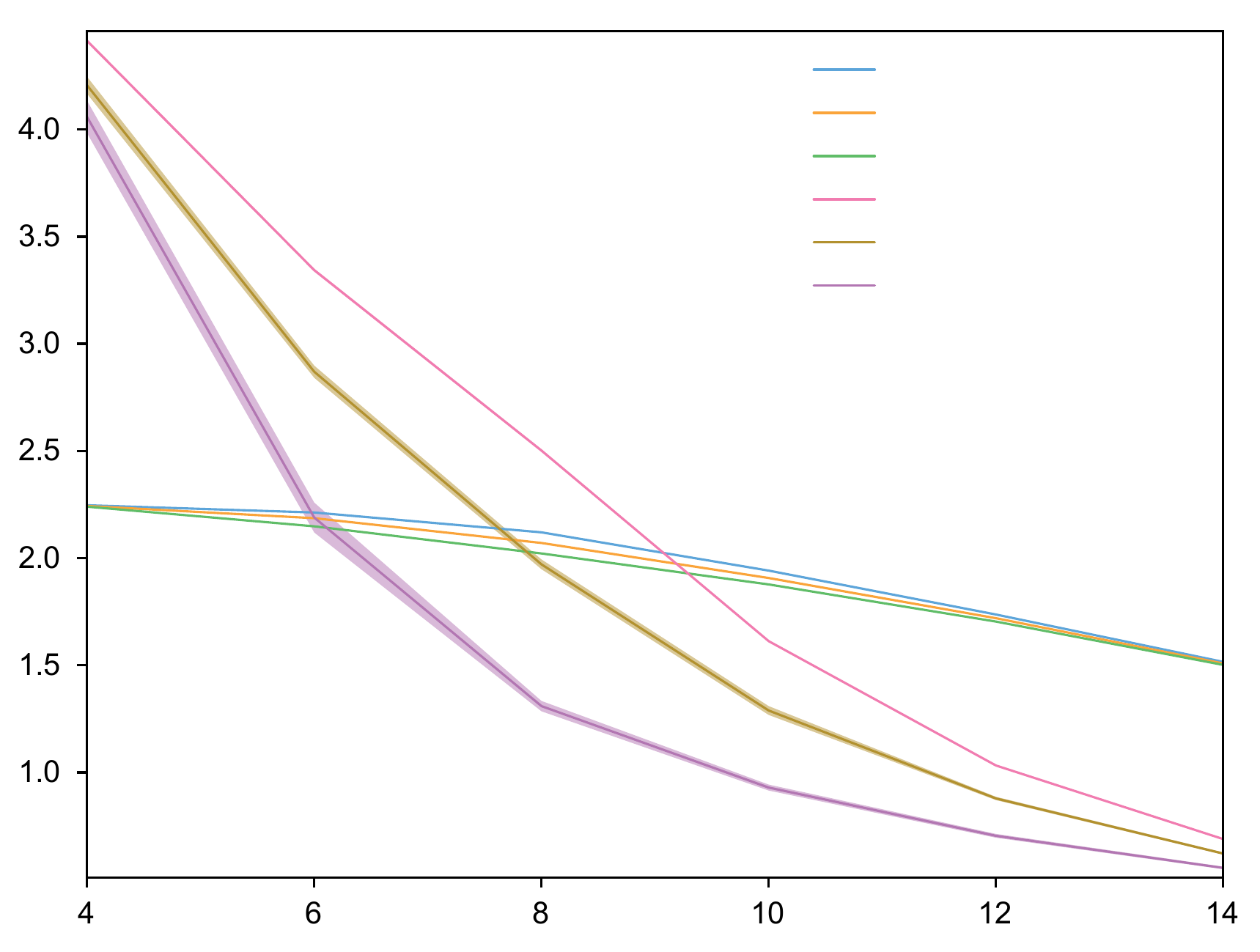}
\put(45, -3){\scalebox{0.75}{regularity, $\rho$}}
\put(-4, 28){\scalebox{0.75}{\rotatebox{90}{{weighted error}}}}
\put(72, 69.5){\scalebox{0.55}{standard proj., $G_\rho \otimes G_\rho$}}
\put(72, 66){\scalebox{0.55}{standard proj., $G_\rho \otimes \tilde G_\rho$}}
\put(72, 62.5){\scalebox{0.55}{standard proj., $\tilde G_\rho \otimes \tilde G_\rho$}}
\put(72, 59){\scalebox{0.55}{debiased proj., $G_\rho \otimes G_\rho$}}
\put(72, 55.5){\scalebox{0.55}{debiased proj., $G_\rho \otimes \tilde G_\rho$}}
\put(72, 52){\scalebox{0.55}{debiased proj., $\tilde G_\rho \otimes \tilde G_\rho$}}
\end{overpic}
\vspace{1mm}
\caption{Weighted error for sampling patterns taken from the adjacency matrices of graph products of
low- and high-spectral gap, regular graphs. Shading indicates a single standard deviation over the draws $\tilde G_\rho$}
\label{fig:spectralgap}
\end{figure}
Figure~\ref{fig:spectralgap} shows that, as expected, the debiased projection algorithm performs better when the spectral gap is smaller.  This is also true of standard projection, although the effect is less pronounced.  As $\rho$ increases (that is, as $\Omega$ becomes denser), the debiased projection algorithm  out-performs the standard projection algorithm.

\section*{Acknowledgements}
We would like to thank the AIM SQuaRE program, where this work was begun.  SF is partially supported by NSF grants DMS-1622134 and DMS-1664803.  DN was supported by NSF CAREER DMS-1348721 and NSF BIGDATA  1740325.   RP is partially supported by a Stanford UAR Major Grant. YP is partially supported by an NSERC Discovery Grant (22R23068), a PIMS CRG 33: High-Dimensional Data Analysis, and a Tier II CRC in Data Science. MW is partially supported by NSF CAREER award CCF-1844628. We would also like to thank Rachel Ward for helpful conversations and for pointing out the potential application in Corollary~\ref{cor:leveragescores}.

\bibliographystyle{plain}
\bibliography{refs}
\end{document}